\pdfoutput=1
\documentclass[12pt]{extarticle}
\usepackage[utf8]{inputenc}
\usepackage{amsmath,amssymb,amsfonts,latexsym}
\usepackage{amsthm}
\usepackage{bm}
\usepackage{graphicx}

\usepackage[margin=1.5cm]{geometry}
\AtBeginDocument{\linespread{0.95}\selectfont}
\usepackage{lineno}
\newcommand*\patchAmsMathEnvironmentForLineno[1]{%
  \expandafter\let\csname old#1\expandafter\endcsname\csname #1\endcsname
  \expandafter\let\csname oldend#1\expandafter\endcsname\csname end#1\endcsname
  \renewenvironment{#1}%
    {\linenomath\csname old#1\endcsname}%
    {\csname oldend#1\endcsname\endlinenomath}}%
\newcommand*\patchBothAmsMathEnvironmentsForLineno[1]{%
  \patchAmsMathEnvironmentForLineno{#1}%
  \patchAmsMathEnvironmentForLineno{#1*}}%
\AtBeginDocument{%
  \patchBothAmsMathEnvironmentsForLineno{equation}%
  \patchBothAmsMathEnvironmentsForLineno{align}%
  \patchBothAmsMathEnvironmentsForLineno{flalign}%
  \patchBothAmsMathEnvironmentsForLineno{alignat}%
  \patchBothAmsMathEnvironmentsForLineno{gather}%
  \patchBothAmsMathEnvironmentsForLineno{multline}%
}
\usepackage{hyperref}
\hypersetup{colorlinks=true,allcolors=blue}

\newtheorem{definition}{Definition}
\newtheorem{theorem}{Theorem}
\newtheorem{lemma}{Lemma}
\newtheorem{corollary}{Corollary}
\newtheorem{remark}{Remark}
\newtheorem{proposition}{Proposition}

\newcommand{\ket}[1]{\lvert #1 \rangle}
\newcommand{\bra}[1]{\langle #1 \rvert}
\newcommand{\braket}[2]{\langle #1 \vert #2 \rangle}
\newcommand{\ketbra}[2]{\lvert #1 \rangle\langle #2 \rvert}

\title{\textbf{Coherent-State Kernels for Composite Quantum Operations}}
\author{Ghasem Asadi Cordshooli\\[2mm]
\small Department of Physics, Yadegar-e-Imam Khomeini Branch (Y.I.C.),\\ \small Islamic Azad University, Tehran, Iran\\
\small E-mail: \texttt{ghascor@iau.ac.ir}}
\date{}

\newcommand{\doi}[1]{\href{https://doi.org/#1}{doi:#1}}
\newcommand{\arxivlink}[1]{\href{https://arxiv.org/abs/#1}{arXiv:#1}}
\begin{document}
\maketitle

\begin{abstract}
This paper develops a unified kernel-based method for computing exact matrix elements of composite quantum operations by inserting resolutions of the identity after individual operators or between successive factors of an operator chain. A single insertion gives an exact scalar identity, and nested insertions between the factors give multiple integral identities for chains of closed operators under explicit Bochner-integrability conditions. The procedure is carried through on polynomial and analytic functions of the ladder operators, on one- and two-mode Gaussian unitaries, and on three composite operations: a Kerr phase followed by squeezing or by a cubic phase gate, a two-mode squeezer followed by a beam splitter, and a heralded cat-state circuit with two beam splitters. At rational Kerr phases the Kerr-squeezed amplitudes reduce to finite sums of squeezed coherent-state amplitudes weighted by quadratic Gauss sums, with an explicit error bound at the remaining phases; the two-mode network gives a bivariate Hermite polynomial; and the heralded circuit gives closed forms for the amplitudes, the heralding probabilities and their extension to an arbitrary passive network fed by one squeezed source. The examples are set against the reach of four established techniques: operator ordering, the Berezin symbol calculus, Lie disentangling and the symplectic formalism. Each technique reaches a subset of the examples directly, while the insertion procedure reaches all of them under the same hypotheses. A comparison with direct Fock summation identifies the regime in which the finite Kerr sums are efficient, and the choice of resolution at each boundary of a chain follows a selection rule set by the action of the adjacent factor. Each closed form agrees with an independent numerical computation.
\end{abstract}


\section{Introduction}
\label{sec:intro}

The resolution of the identity, an integral of projectors onto an overcomplete family of states, enters most calculations in quantum theory as a passive step that changes the representation and disappears once the calculation ends. Its coherent-state form goes back to the minimum-uncertainty wave packets of Schr\"odinger \cite{Schrodinger1926}, and it became a working tool with Glauber's coherent states of the radiation field \cite{Glauber1963,Glauber1963b} and Sudarshan's diagonal representation of density operators \cite{Sudarshan1963}. Klauder formulated continuous representations axiomatically, as overcomplete families that resolve the identity and support a phase-space form of the path integral \cite{Klauder1963,Klauder1963_2}, and Bargmann identified the underlying Hilbert space of entire functions, in which the resolution acts as a reproducing kernel \cite{Bargmann1961}. Since the resolution is an integral over phase space, every calculation that inserts it is an integration problem, and the established techniques of quantum optics are, at their core, methods for evaluating such integrals.

Four families of techniques dominate the evaluation of these integrals, and Table~\ref{tab:reach} compares them on the examples of this paper. Ordering techniques bring products of ladder operators into normal, antinormal or Weyl order; inside the ordering symbol the operators commute, so integrals over ket--bra projectors reduce to ordinary Gaussian integrals. The technique of integration within ordered products gives in this way the normally ordered squeeze operators \cite{FanZaidiKlauder1987} and a large family of operator identities \cite{Fan2003,FanYuan2010}, and the $s$-parameterized ordering of Cahill and Glauber interpolates between the three orders \cite{CahillGlauber1969}. Berezin's calculus encodes an operator by its covariant and contravariant symbols on phase space, represents operator products by integral compositions of symbols, and builds a phase-space functional integral from successive approximations to the exponential of an operator \cite{Berezin1972}. Lie-algebraic disentangling factorizes exponentials of the generators of a finite-dimensional Lie algebra into ordered products \cite{Truax1985}, group-orbit coherent states extend this structure to general Lie groups \cite{Perelomov1986}, and the Magnus expansion extends it to time-dependent generators \cite{Blanes2009}. For Gaussian operations, the symplectic formalism represents a Gaussian unitary by a symplectic matrix, and its Fock amplitudes follow as hafnians and loop hafnians \cite{Hamilton2017,Kruse2019,Bjorklund2019,Quesada2019} or from recursion relations \cite{MiattoQuesada2020}; this route underlies the analysis of Gaussian boson sampling \cite{Madsen2022}.

Each technique reaches its limit where its structural assumption ends. Ordering techniques and Lie disentangling rest on generators of at most quadratic order in the ladder operators; for a quartic generator, such as the Kerr Hamiltonian proportional to the square of the photon number, repeated commutators span an infinite-dimensional algebra and the ordered form becomes an infinite series. Non-monomial orderings, such as the Weyl and $s$-orderings, admit many distinct schemes for the same ordered value \cite{Diosi2017}. Berezin's covariant symbol exists for every bounded operator, whereas the existence of a contravariant symbol for a given operator is a separate question \cite{Berezin1972}. The Magnus series converges under a bound on the integrated norm of the generator \cite{Blanes2009}. The symplectic and hafnian route covers Gaussian operations, and non-Gaussian elements enter through decompositions into sums of Gaussian or coherent-state terms \cite{Chabaud2020,Bourassa2021,MarshallAnand2023,DiasKonig2024,Guseynov2026}, at a cost that grows with the number of terms. For the Kerr gate, the fractional-revival decomposition writes the evolution at rational phases as a finite sum of phase rotations weighted by quadratic Gauss sums \cite{YurkeStoler1986,TaraAgarwalChaturvedi1993,WuGuo1996,Spiridonov2022}, and Diophantine approximation controls the error at the remaining phases, so that every term reduces to Gaussian operations \cite{Upreti2026}.

The four techniques are closely related, since each evaluates integrals generated by coherent-state resolutions of the identity. The normally ordered form of an operator is its covariant, or Wick, symbol, and the antinormally ordered form corresponds to its contravariant, or anti-Wick, symbol \cite{Berezin1972,CahillGlauber1969}; integration within ordered products therefore realizes the composition of Berezin symbols for bosonic operators. For Gaussian operators, the normally ordered forms obtained by ordering techniques coincide with the factorizations produced by Lie disentangling \cite{FanZaidiKlauder1987,Truax1985}, and both encode the symplectic matrix of the transformation, from which the hafnian formulas follow \cite{Quesada2019}. The orderings are linked to one another by a generalized Wick theorem that follows from the Baker--Campbell--Hausdorff identity \cite{Diosi2017}. The coherent-state path integral is the continuum limit of repeated insertions of the resolution between short evolution steps \cite{Klauder1963_2,Berezin1972}.

In this paper the resolution of the identity is used as a method for deriving exact matrix elements of operator chains, extending its pedagogical form for single operators in Ref.~\cite{Asadi2026AJP}. The number of insertions is finite and fixed by the factorization of the chain, and each insertion is exact, in contrast with the continuum limit that defines the path integral. A single insertion between an operator and a projection gives an exact scalar identity that follows from the weak form of the resolution alone. Insertions between the successive factors of a chain of $N$ closed operators give an $(N-1)$-fold integral identity under explicit Bochner-integrability conditions, with the same hypotheses for every row of Table~\ref{tab:reach}. The examples include three composite operations: a Kerr phase followed by squeezing, which reduces to Gauss-sum combinations of squeezed coherent-state amplitudes at rational phases, with an explicit error bound at the remaining phases; a two-mode squeezer followed by a beam splitter, whose Fock amplitudes form a bivariate Hermite polynomial; and the heralded cat-state circuit of Ref.~\cite{Podoshvedov2023}, extended to an arbitrary passive network fed by one squeezed source. A comparison with direct Fock summation shows that the finite Kerr sums are the efficient route at rational phases with few terms, and Section~\ref{sec:choice} formulates the choice of resolution at each boundary of a chain as a selection rule. Every closed form is checked against an independent numerical computation. The fractional-revival decomposition of rational Kerr phases, the single- and two-mode squeezing kernels and the beam-splitter map are established results, and the amplitudes of Gaussian networks are special cases of the loop-hafnian formulas. The contributions of this paper are the analytic conditions under which the nested insertion holds, the explicit closed forms of the composite amplitudes obtained from the kernels of their factors, the amplitude-level error bound at irrational Kerr phases with its comparison against direct Fock summation, and the combination of the Kerr decomposition with the Airy closed form of the cubic phase gate \cite{Asadi2026cubic}.

Section~\ref{sec:framework} defines the resolution of the identity and chains of operators, and Section~\ref{sec:assumptions} states the assumptions together with the two exchange lemmas. Section~\ref{sec:factory} gives the single insertion and the nested theorem. Section~\ref{sec:examples} works through the examples in order of increasing structural complexity and compares the closed forms with numerical computations. Section~\ref{sec:discussion} compares the reach of the procedure with the established techniques, examines the choice of resolution, and states its scope and prospective applications.

%%%%%%%%%%%%%%%%%%%%%%%%%%%%%%%%%%%%%%%%%%%%%%%%%%%%%%%%%%%%%%%%%%%%%%
\section{Definitions}
\label{sec:framework}

This section fixes the two objects used throughout: a resolution of the identity and a chain of operators.

\begin{definition}[Resolution of the Identity]
Let $\mathcal H$ be a separable Hilbert space. A continuous family $\{\ket{\gamma}\}_{\gamma\in\Omega}$ with $\|\ket{\gamma}\|=1$ is said to resolve the identity if
\begin{equation}
I=\int_{\Omega} d\mu(\gamma)\,\ketbra{\gamma}{\gamma},
\label{ROI}
\end{equation}
where $d\mu(\gamma)$ is a positive measure on $\Omega$, and the integral is understood weakly
\[
\langle\phi|I|\chi\rangle
=
\int_\Omega d\mu(\gamma)\,
\langle\phi|\gamma\rangle
\langle\gamma|\chi\rangle,
\qquad
\forall\,\phi,\chi\in\mathcal H.
\]
\end{definition}

\begin{definition}[Chain of Linear Operators]
Let the ordered product of operators be denoted by $\prod_{k=1}^N \mathcal{F}_k \equiv \mathcal{F}_N \cdots \mathcal{F}_1$, where each constituent operator $\mathcal{F}_k$ is a densely defined, closed linear operator on a separable Hilbert space $\mathcal{H}$
\begin{equation}
\mathcal{F}_k: \mathcal{D}(\mathcal{F}_k) \subseteq \mathcal{H} \longrightarrow \mathcal{H}.
\end{equation}
When $N=1$, this product collapses to a single operator, denoted simply by $\mathcal{F}$.
\label{Domain}
\end{definition}
\begin{remark}[Typical Operator Functions]
The present framework includes, in particular, normal-ordered polynomials, analytic operator functions, operator products, powers, and convergent operator series.
\end{remark}

\begin{remark}[Nested insertion]
\label{rem:nested}
For every chain of densely defined operators $\mathcal F_1,\ldots,\mathcal F_N$ and every $\ket{\psi}\in \mathcal D(\mathcal F_N\cdots\mathcal F_1)$, the identity
\begin{equation}
\mathcal F_N\cdots\mathcal F_1\ket{\psi}
=
\mathcal F_N\cdots\mathcal F_1\ket{\psi}
\label{eq:chain}
\end{equation}
holds trivially. The term ``nested insertion'' denotes the sequential insertion of a resolution of the identity \emph{between} successive factors of an operator product, which produces nested integration kernels. Proposition~\ref{cor:single} treats one insertion followed by a projection onto an orthonormal basis, and Theorem~\ref{thm:telescope} treats a chain of factors.
\end{remark}

%%%%%%%%%%%%%%%%%%%%%%%%%%%%%%%%%%%%%%%%%%%%%%%%%%%%%%%%%%%%%%%%%%%%%%
\section{Assumptions}
\label{sec:assumptions}

This section states the hypotheses used throughout and the two exchange lemmas that follow from them. Throughout this paper, the following are assumed:

\begin{enumerate}
\item The Hilbert space $\mathcal H$ is separable, with a fixed orthonormal basis $\{\ket n\}_{n=0}^\infty$.

\item The continuous resolution of the identity (\ref{ROI}) holds, with $\|\ket\gamma\|=1$ for $\mu$-almost every $\gamma\in\Omega$.

\item The target vector satisfies $\ket\psi\in\mathcal D(\mathcal F)$ and $\mathcal F\ket\psi\in\mathcal H$, ensuring that the projection $\langle n|\mathcal F\psi\rangle$ is well defined for every $n$.

\item \textbf{Bochner integrability.} The vector-valued map $\gamma\mapsto \ket\gamma\,\langle\gamma|\mathcal F\psi\rangle$ is Bochner integrable with respect to $\mu$, which is equivalent to the condition
\begin{equation}
\int_\Omega d\mu(\gamma)\,\bigl|\langle\gamma|\mathcal F\psi\rangle\bigr| \;<\;\infty.
\label{BochnerCond}
\end{equation}

\item \textbf{Domain and Bochner integrability for compositions.} Let $\mathcal G$ be a closed, densely defined operator on $\mathcal H$, and let $\ket\phi\in\mathcal H$ be a vector for which the Bochner integral $\int_\Omega d\mu(\gamma)\,\ket\gamma\langle\gamma|\phi\rangle$ exists (i.e., $\ket\phi$ satisfies Assumption~4 with $\mathcal F\ket\psi$ replaced by $\ket\phi$). It is assumed that:
\begin{enumerate}
\item[(a)] $\ket\gamma\in\mathcal D(\mathcal G)$ for $\mu$-almost every $\gamma\in\Omega$;
\item[(b)] the vector-valued map $\gamma\mapsto \mathcal G\ket\gamma\,\langle\gamma|\phi\rangle$ is Bochner integrable with respect to $\mu$, i.e.,
\begin{equation}
\int_\Omega d\mu(\gamma)\,\bigl\|\mathcal G\ket\gamma\bigr\|\,\bigl|\langle\gamma|\phi\rangle\bigr| \;<\;\infty.
\label{BochnerCondG}
\end{equation}
\end{enumerate}
\end{enumerate}

\begin{lemma}[Exchange of integral and projection]
\label{lem:exchange}
Under Assumption~4, the vector-valued integral $\int_\Omega d\mu(\gamma)\,\ket\gamma\langle\gamma|\mathcal F\psi\rangle$ exists as a Bochner integral in $\mathcal H$, and for every basis vector $\ket n$,
\begin{equation}
\left\langle n\left|\int_{\Omega} d\mu(\gamma)\, \ket\gamma\langle\gamma|\mathcal F\psi\rangle\right.\right\rangle
=\int_{\Omega} d\mu(\gamma)\,\langle n|\gamma\rangle\langle\gamma|\mathcal F\psi\rangle.
\label{Exchange}
\end{equation}
\end{lemma}

\begin{proof}
Since $\|\ket\gamma\|=1$, the condition (\ref{BochnerCond}) is precisely the statement $\int_\Omega d\mu(\gamma)\,\|\ket\gamma\langle\gamma|\mathcal F\psi\rangle\|<\infty$, which constitutes the standard criterion for Bochner integrability of the map $\gamma\mapsto\ket\gamma\langle\gamma|\mathcal F\psi\rangle$ \cite{DiestelUhl1977}. A fundamental property of the Bochner integral is that any bounded linear functional $\varphi$ on $\mathcal H$ commutes with the integration: $\varphi\bigl(\int_\Omega f\,d\mu\bigr)=\int_\Omega \varphi(f)\,d\mu$. Setting $\varphi=\langle n|\cdot\rangle$, which is bounded since $\|\,\langle n|\cdot\rangle\,\|=1$, directly yields (\ref{Exchange}).
\end{proof}

Assumption~4 enters the vector-valued statement \eqref{Exchange}; the scalar identity \eqref{SingleInsert} follows from Proposition~\ref{cor:single} alone. Assumption~4 is a concrete and checkable hypothesis: the continuous symbol $\gamma\mapsto\langle\gamma|\mathcal F\psi\rangle$ of the vector $\mathcal F\ket\psi$ is absolutely integrable on $(\Omega,\mu)$. In the applications of Sec.~\ref{sec:examples}, this reduces to the integrability of a Gaussian or Gaussian-dominated function of $\gamma$.

\begin{lemma}[Exchange of a closed operator with the Bochner integral]
\label{lem:exchangeG}
Under Assumption~5, the vector $\ket\phi$ lies in the domain $\mathcal D(\mathcal G)$, the vector-valued integral $\int_\Omega d\mu(\gamma)\,\mathcal G\ket\gamma\,\langle\gamma|\phi\rangle$ exists as a Bochner integral in $\mathcal H$, and
\begin{equation}
\mathcal G\int_\Omega d\mu(\gamma)\,\ket\gamma\langle\gamma|\phi\rangle
= \int_\Omega d\mu(\gamma)\,\mathcal G\ket\gamma\,\langle\gamma|\phi\rangle.
\label{ExchangeG}
\end{equation}
\end{lemma}

\begin{proof}
By hypothesis (b), condition (\ref{BochnerCondG}) is the Bochner-integrability criterion \cite{DiestelUhl1977} for the map $\gamma\mapsto\mathcal G\ket\gamma\,\langle\gamma|\phi\rangle$, ensuring the existence of the right-hand side of (\ref{ExchangeG}) as a well-defined vector in $\mathcal H$. To show that this vector equals $\mathcal G$ applied to the integral of the unacted states, the closed-graph criterion for Bochner integrals is invoked: if $\mathcal G$ is a closed operator, and both $\gamma\mapsto\ket\gamma\langle\gamma|\phi\rangle$ and $\gamma\mapsto\mathcal G\ket\gamma\langle\gamma|\phi\rangle$ are Bochner integrable with $\ket\gamma\in\mathcal D(\mathcal G)$ almost everywhere, then the integral of the first map lies in $\mathcal D(\mathcal G)$ and satisfies (\ref{ExchangeG}) \cite{DiestelUhl1977, Hille1957}. Assumptions~5(a) and 5(b) provide exactly these conditions, completing the proof.
\end{proof}

The exchange \eqref{ExchangeG} rests on three ingredients together: the closedness of $\mathcal G$, the domain condition (a) and the integrability condition (b). Assumption~5 states all three explicitly.

%%%%%%%%%%%%%%%%%%%%%%%%%%%%%%%%%%%%%%%%%%%%%%%%%%%%%%%%%%%%%%%%%%%%%%
\section{The Nested Integration Framework}
\label{sec:factory}

This section states the mechanism that converts chains of operator products into exact multi-fold integral identities.

The construction has two layers. A resolution of the identity inserted once, after the operator and before the projection, produces an exact scalar identity; this is the master identity used in all single-operator examples below, and it follows from Definition~1 alone. Inserting resolutions \emph{between} the factors of a product is a second and stronger step: it moves a closed operator through a vector-valued integral, and this is where the integrability hypotheses enter. The two layers are stated in that order.

\subsection{A Single Insertion}

This subsection shows that one insertion gives an exact scalar identity. For a positive resolution of the identity the inserted integral converges absolutely for any two vectors of $\mathcal H$, by the Cauchy--Schwarz inequality in $L^2(\Omega,\mu)$,
\begin{equation}
\int_\Omega d\mu(\gamma)\,\bigl|\langle n|\gamma\rangle\langle\gamma|\phi\rangle\bigr|
\;\le\;
\Bigl(\int_\Omega d\mu\,|\langle n|\gamma\rangle|^2\Bigr)^{1/2}
\Bigl(\int_\Omega d\mu\,|\langle\gamma|\phi\rangle|^2\Bigr)^{1/2}
= \|\ket n\|\,\|\ket\phi\| .
\label{eq:CS}
\end{equation}

\begin{proposition}[Single-insertion identity]
\label{cor:single}
Let $\mathcal F$ be a linear operator with $\ket\psi\in\mathcal D(\mathcal F)$, so that $\mathcal F\ket\psi\in\mathcal H$ (Assumption~3). Then the identity $\mathcal F\ket\psi=I\,\mathcal F\ket\psi$ gives, for every basis index $n$, the exact identity
\begin{equation}
\langle n|\mathcal F|\psi\rangle
=\int_{\Omega} d\mu(\gamma)\,\langle n|\gamma\rangle\langle\gamma|\mathcal F|\psi\rangle ,
\label{SingleInsert}
\end{equation}
and the integral converges absolutely.
\end{proposition}

\begin{proof}
Apply the weak form of (\ref{ROI}), as stated in Definition~1, to the pair of vectors $\ket n$ and $\mathcal F\ket\psi$. Absolute convergence is \eqref{eq:CS} with $\ket\phi=\mathcal F\ket\psi$.
\end{proof}

Equation \eqref{SingleInsert} is a statement about scalars, and it follows from Definition~1 alone. Assumption~4 and the closedness of $\mathcal F$ enter when the insertion holds at the level of vectors, as in the composition of several operators treated in the theorem below.

\subsection{The Main Result}

This subsection extends the insertion to chains of operators.

\begin{theorem}[Nested Integral Generator Theorem]
\label{thm:telescope}
Sequentially inserting resolutions of the identity between the factors of an operator product generates an exact $(N-1)$-fold integral representation of the corresponding compound matrix element. More precisely, let $N\ge2$ and let $\mathcal F_1, \dots, \mathcal F_N$ be a finite chain of closed, densely defined linear operators on a separable Hilbert space $\mathcal H$. Let $\ket\psi \in \mathcal D(\mathcal F_1)$, and define the intermediate state vectors sequentially by
\begin{equation}
\ket{\phi_0} := \ket\psi, \quad \text{and} \quad \ket{\phi_k} := \mathcal F_k \ket{\phi_{k-1}} \quad \text{for } k=1,\dots,N-1.
\end{equation}
Suppose that at each stage $k$:
\begin{itemize}
    \item[(i)] The vector \(\ket{\phi_k}\) satisfies the Bochner integrability condition (Assumption~4) with respect to a resolution of the identity \(\{\ket{\gamma_k}\}_{\gamma_k\in\Omega_k}\) satisfying (\ref{ROI});
    \item[(ii)] The operator-vector pair \((\mathcal F_{k+1}, \ket{\phi_k})\) satisfies the domain and exchange conditions (Assumption~5).
\end{itemize}
Then, the composite domain is preserved, i.e., $\ket\psi \in \mathcal D(\mathcal F_N \cdots \mathcal F_1)$, and for every basis index $n$, the following exact $(N-1)$-fold nested integral identity holds
\begin{equation}
\boxed{
\langle n|\mathcal F_N\cdots\mathcal F_1|\psi\rangle
=\int_{\Omega_{N-1}} \!\!\!\! d\mu(\gamma_{N-1})\cdots\int_{\Omega_1} \!\! d\mu(\gamma_1)\;
\langle n|\mathcal F_N|\gamma_{N-1}\rangle\langle\gamma_{N-1}|\mathcal F_{N-1}|\gamma_{N-2}\rangle\cdots\langle\gamma_1|\mathcal F_1|\psi\rangle
}
\label{TelescopeN}
\end{equation}
\end{theorem}

\begin{proof}
The proof proceeds by induction on the chain length $N$.

\paragraph{Base case ($N=2$).} Let $\ket\phi := \mathcal F_1\ket\psi$. By hypothesis (i) and the direct application of Lemma~\ref{lem:exchange} to $\ket\phi$, this gives
\begin{equation}
\mathcal F_1\ket\psi = \ket\phi = \int_{\Omega_1} d\mu(\gamma_1)\,\ket{\gamma_1}\langle\gamma_1|\phi\rangle.
\label{Base1}
\end{equation}
Applying the closed operator $\mathcal F_2$ to both sides, hypothesis (ii) and Lemma~\ref{lem:exchangeG} justify exchanging the operator with the Bochner integral
\begin{equation}
\mathcal F_2\ket\phi = \int_{\Omega_1} d\mu(\gamma_1)\,\mathcal F_2\ket{\gamma_1}\,\langle\gamma_1|\phi\rangle.
\label{Base2}
\end{equation}
This confirms that $\ket\psi \in \mathcal D(\mathcal F_2\mathcal F_1)$. Finally, taking the inner product with the basis vector $\bra n$ and invoking Lemma~\ref{lem:exchange} once more yields the $N=2$ identity
\begin{equation}
\langle n|\mathcal F_2\mathcal F_1|\psi\rangle = \int_{\Omega_1} d\mu(\gamma_1)\,\langle n|\mathcal F_2|\gamma_1\rangle\,\langle\gamma_1|\mathcal F_1|\psi\rangle.
\end{equation}

\paragraph{Inductive step.} Assume the theorem holds for a chain of length $N-1$, meaning the state $\ket{\phi_{N-1}} := \mathcal F_{N-1}\cdots\mathcal F_1\ket\psi$ satisfies
\begin{equation}
\langle\gamma_{N-1}|\phi_{N-1}\rangle
= \int_{\Omega_{N-2}} \!\!\!\! d\mu(\gamma_{N-2})\cdots\int_{\Omega_1} \!\! d\mu(\gamma_1)\,
\langle\gamma_{N-1}|\mathcal F_{N-1}|\gamma_{N-2}\rangle\cdots\langle\gamma_1|\mathcal F_1|\psi\rangle.
\label{InductiveHyp}
\end{equation}
By hypothesis, the pair $(\mathcal F_N, \ket{\phi_{N-1}})$ satisfies Assumption~5. Applying the base-case result to the product $\mathcal F_N \ket{\phi_{N-1}}$ gives
\begin{equation}
\langle n|\mathcal F_N|\phi_{N-1}\rangle = \int_{\Omega_{N-1}} d\mu(\gamma_{N-1})\,\langle n|\mathcal F_N|\gamma_{N-1}\rangle\langle\gamma_{N-1}|\phi_{N-1}\rangle.
\end{equation}
Substituting (\ref{InductiveHyp}) into this expression completes the induction and proves (\ref{TelescopeN}) for all $N$.
\end{proof}

\subsection{Relation between the two layers}

This subsection compares the two layers. The two statements insert the resolution at different places and complement each other. Proposition~\ref{cor:single} inserts it \emph{after} the operator, between $\mathcal F\ket\psi$ and the bra, and gives one integral whatever the length of the chain. Theorem~\ref{thm:telescope} inserts it \emph{between} successive factors, and for a chain of $N$ factors gives $N-1$ integrals, each step of which uses the vector-level hypotheses. For $N=2$ the theorem gives the single insertion between $\mathcal F_2$ and $\mathcal F_1$, the form used in most of the examples below.

\section{Examples}
\label{sec:examples}

This section applies the procedure to physical examples in increasing order of structural complexity. Single-operator cases come first, where the construction reduces to the single-insertion identity of Proposition~\ref{cor:single}; compound networks with $N\ge2$ factors follow, where the nested insertions of Theorem~\ref{thm:telescope} enter.

Throughout, the canonical Glauber coherent states are employed
\[
\ket{\alpha}
=
e^{-|\alpha|^2/2}
\sum_{n=0}^{\infty}
\frac{\alpha^n}{\sqrt{n!}}\ket{n},
\]
which satisfy the resolution of the identity
\[
I
=
\int \frac{d^2\beta}{\pi}\,\ketbra{\beta}{\beta},
\qquad
\braket{\beta}{\alpha}
=
\exp\!\left[-\tfrac12|\alpha|^2-\tfrac12|\beta|^2+\beta^*\alpha\right].
\]
For convenience, the target state is chosen to be a coherent state of the corresponding Hilbert space throughout the examples. Accordingly, single-mode operators act on the coherent state $\ket{\alpha}$ (i.e., $\ket{\psi}=\ket{\alpha}$), whereas two-mode operators act on the product coherent state $\ket{\alpha,\lambda}=\ket{\alpha}\otimes\ket{\lambda}$.

Assumptions~4 and 5 follow uniformly for the Gaussian examples. For a pure Gaussian state $\ket\varphi$ of $M$ modes, the function $\gamma\mapsto|\braket\gamma\varphi|^2$ is $\pi^M$ times the Husimi function of $\ket\varphi$, a Gaussian with positive-definite covariance matrix, so $|\braket\gamma\varphi|$ is integrable over $\mathbb C^M$, and so is its product with any polynomial in $\gamma$ and $\gamma^*$. For the Gaussian unitaries and the polynomial functions of the ladder operators below, every vector $\mathcal F\ket\psi$ is a pure Gaussian state or a polynomial in the ladder operators applied to one, which gives Assumption~4; for a unitary factor $\|U\ket\gamma\|=1$, and Assumption~5 reduces to Assumption~4. The analytic functions of Sec.~\ref{sec:consolidated} carry the growth condition stated there, and the Kerr example is checked explicitly in Sec.~\ref{sec:kerr}.

\subsection{Identity Operator}

The simplest case fixes the Gaussian moment identity used throughout. For $\mathcal F=I$, the single-insertion identity (\ref{SingleInsert}) gives $\braket{n}{\alpha}=\int\frac{d^2\beta}{\pi}\braket{n}{\beta}\braket{\beta}{\alpha}$. Expanding both coherent states in the Fock basis and canceling the common factor $e^{-|\alpha|^2/2}/\sqrt{n!}$ yields the familiar Gaussian moment identity
\begin{equation}
\int d^2\beta\, e^{-|\beta|^2+\beta^*\alpha}\,\beta^n = \pi\,\alpha^n.
\tag{MI-1}
\label{eq:MI1}
\end{equation}

\subsection{Elementary, Polynomial, and Analytic Operator Functions (Consolidated)}
\label{sec:consolidated}

This subsection treats all functions of the ladder operators with one computation. The coherent-state bra and ket transform simply under the ladder operators
\begin{equation}
a\ket\alpha=\alpha\ket\alpha
\qquad\Longrightarrow\qquad
\bra\beta\,a^k\,\ket\alpha=\alpha^k\braket\beta\alpha,
\qquad
\bra\beta\,(a^\dagger)^k = (\beta^*)^k\bra\beta,
\label{eq:symbol-rule}
\end{equation}
where the second relation is the adjoint of $a^k\ket\beta=\beta^k\ket\beta$. Consequently, for any normal-ordered polynomial $\mathcal F=(a^\dagger)^k a^l$
\begin{equation}
\bra\beta\,\mathcal F\,\ket\alpha = (\beta^*)^k\alpha^l\,\braket\beta\alpha ,
\end{equation}
and the single-insertion identity (\ref{SingleInsert}) becomes, after using $\braket n\beta=e^{-|\beta|^2/2}\beta^n/\sqrt{n!}$ and canceling the common prefactor $e^{-|\alpha|^2/2}/\sqrt{n!}$
\begin{equation}
\boxed{\;
\int d^2\beta\, e^{-|\beta|^2+\beta^*\alpha}\,\beta^n(\beta^*)^k
=
\pi\,\frac{n!}{(n-k)!}\,\alpha^{n-k},
\qquad n\ge k,
\;}
\tag{MI-2}
\label{eq:MI2}
\end{equation}
(and zero for $n<k$), so that $\langle n|(a^\dagger)^k a^l|\alpha\rangle = \alpha^l\,\dfrac{n!}{(n-k)!\sqrt{n!}}\,\alpha^{n-k}\,e^{-|\alpha|^2/2}$ follows directly from \eqref{eq:MI2}. This single computation subsumes, as the special cases $(k,l)=(0,1)$, $(1,0)$, and $(1,1)$, the annihilation operator $a$, the creation operator $a^\dagger$, and the number operator $N=a^\dagger a$, respectively; the general polynomial $\mathcal F=\sum_{k,l}c_{kl}(a^\dagger)^ka^l$ follows by linearity.

\paragraph{Two distinct analytic generalizations.}
Equation \eqref{eq:MI2} admits two different analytic completions, according to whether the auxiliary function is holomorphic in $\beta$ or in $\beta^*$. The two completions have different right-hand sides and serve different purposes.

\emph{(i) Functions of $\beta$ (the Bargmann reproducing-kernel identity).} For $f$ entire and of suitable growth, term-by-term integration using $\int d^2\beta\,\beta^{m} e^{-|\beta|^2+\beta^*\alpha}=\pi\alpha^m$ on each Taylor coefficient of $f$ gives the standard Segal--Bargmann reproducing property
\begin{equation}
\boxed{\;
\int d^2\beta\, e^{-|\beta|^2+\beta^*\alpha}\,\beta^n f(\beta)
=
\pi\,\alpha^n f(\alpha) .
\;}
\tag{MI-3a}
\label{eq:MI3a}
\end{equation}
Setting $f(\beta)=\beta^k$ gives $\pi\alpha^{n+k}$, which differs from the right-hand side of \eqref{eq:MI2}. The difference is structural: $(\beta^*)^k$ is anti-holomorphic in $\beta$, whereas $f(\beta)$ is holomorphic.

\emph{(ii) Functions of $\beta^*$ (the operator-symbol identity).} This is the case actually required to unify $a$, $a^\dagger$, $N$, and general analytic functions of $a^\dagger$ into a single formula, since it is precisely the bra-side relation \eqref{eq:symbol-rule} --- $\bra\beta\,f(a^\dagger)\,\ket\alpha=f(\beta^*)\braket\beta\alpha$, valid for $f$ built from normal-ordered powers through its convergent Taylor expansion --- that produces an integrand depending on $\beta^*$. Expanding $f(\beta^*)=\sum_k c_k(\beta^*)^k$ and applying \eqref{eq:MI2} term by term, the sum $\sum_k c_k\,n!/(n-k)!\,\alpha^{n-k}$ is exactly what one obtains by letting the differential operator $f(\partial_\alpha)$ act on $\alpha^n$, since $\partial_\alpha^k\alpha^n=n!/(n-k)!\,\alpha^{n-k}$. This gives
\begin{equation}
\boxed{\;
\int d^2\beta\, e^{-|\beta|^2+\beta^*\alpha}\,\beta^n f(\beta^*)
=
\pi\, f(\partial_\alpha)\,\alpha^n ,
\;}
\tag{MI-3b}
\label{eq:MI3b}
\end{equation}
which now correctly contains \eqref{eq:MI1} as the case $f\equiv1$ and \eqref{eq:MI2} as the case $f(\beta^*)=(\beta^*)^k$, since $\partial_\alpha^k\alpha^n=n!/(n-k)!\,\alpha^{n-k}$ reproduces \eqref{eq:MI2} exactly. The right-hand side of \eqref{eq:MI3b} is the action of the operator $f(\partial_\alpha)$ on the monomial $\alpha^n$, whereas in \eqref{eq:MI3a} the function $f$ is evaluated at $\alpha$; the two differ as soon as $f$ has a non-constant term. The same substitution extends to nested functions $g(f(A))$ (with $A=a^\dagger$) via
\begin{equation}
\int d^2\beta\,e^{-|\beta|^2+\beta^*\alpha}\beta^n\,g(f(\beta^*))
=
\pi\,(g\circ f)(\partial_\alpha)\,\alpha^n .
\end{equation}

Both \eqref{eq:MI3a} and \eqref{eq:MI3b} are used elsewhere in this paper: \eqref{eq:MI3a} is a standalone fact about entire functions in the Bargmann representation, while \eqref{eq:MI3b} is the identity underlying the treatment of general normal-ordered Hamiltonians in Sec.~\ref{sec:hamiltonians}.

\paragraph{The displacement operator.}
Finally, for the displacement operator $D(\gamma)=\exp(\gamma a^\dagger-\gamma^*a)$, using
\[
D(\gamma)\ket\alpha
=
e^{(\gamma\alpha^*-\gamma^*\alpha)/2}\ket{\alpha+\gamma},
\qquad
\bra\beta D(\gamma)
=
e^{(\gamma\beta^*-\gamma^*\beta)/2}\bra{\beta-\gamma},
\]
(the second relation obtained by taking the adjoint of $D(\gamma)^\dagger\ket\beta=D(-\gamma)\ket\beta=e^{(-\gamma\beta^*+\gamma^*\beta)/2}\ket{\beta-\gamma}$), the master identity reduces, after cancellation of the common phase and Gaussian prefactors, to \eqref{eq:MI1} with $\alpha\to\alpha+\gamma$
\begin{equation}
\int d^2\beta\, e^{-|\beta|^2+\beta^*(\alpha+\gamma)}\,\beta^n
=
\pi(\alpha+\gamma)^n .
\end{equation}
This consolidates the six elementary single-mode examples into the two master formulas \eqref{eq:MI3a} and \eqref{eq:MI3b}, of which they are all special cases.

\paragraph{Other routes.} For these operator functions the normal-ordering techniques \cite{FanZaidiKlauder1987,Fan2003,FanYuan2010} and the Berezin symbol calculus \cite{Berezin1972} give the same results directly. The symplectic formalism applies to Gaussian unitaries, and functions such as $e^{z a^\dagger}$ with $z\in\mathbb C$ lie outside this class.

\subsection{Single-Mode Squeezing Operator}
\label{sec:singlemode-squeeze}

This subsection derives the coherent-state kernel of the single-mode squeezer and its Fock amplitudes. For the single-mode squeezing operator
\[
S(\xi)=\exp\!\left[\frac12\left(\xi^*a^2-\xi a^{\dagger2}\right)\right],
\qquad
\xi=re^{i\theta},
\]
acting directly on $\ket\alpha$, the Bogoliubov transformation is used
\[
S^\dagger(\xi)\,a\,S(\xi)=\mu a-\nu a^\dagger,
\qquad
\mu=\cosh r,\quad \nu=e^{i\theta}\sinh r,
\qquad
|\mu|^2-|\nu|^2=1,
\]
which follows directly from the Baker--Campbell--Hausdorff expansion of $e^{-G}ae^{G}$ with $G=\tfrac12(\xi^*a^2-\xi a^{\dagger2})$\cite{Yuen1976,CavesSchumaker1985I}.

The matrix element $\bra\beta S(\xi)\ket\alpha$ is evaluated with the normally ordered disentangling of the squeeze operator \cite{CavesSchumaker1985I,BarnettRadmore}
\begin{equation}
S(\xi)=\exp\!\left[-\frac{\sigma}{2}a^{\dagger2}\right]
\exp\!\left[-\ln\mu\left(N+\tfrac12\right)\right]
\exp\!\left[\frac{\tau}{2}a^{2}\right],
\qquad
\sigma\equiv\frac{\nu}{\mu},\quad
\tau\equiv\frac{\nu^*}{\mu},
\label{eq:disentangle}
\end{equation}
The rightmost factor acts on the ket through $e^{\tau a^2/2}\ket\alpha=e^{\tau\alpha^2/2}\ket\alpha$, the leftmost factor acts on the bra through $\bra\beta a^\dagger=\beta^*\bra\beta$, and the middle factor follows from the Fock-basis resummation
\begin{equation}
\bra\beta\, e^{-\ln\mu(N+1/2)}=\mu^{-1/2}\exp\!\left[-\frac12\left(1-\frac1{\mu^2}\right)|\beta|^2\right]\left\langle\frac\beta\mu\right|
\label{eq:number-op-action}
\end{equation}
for the middle factor, one finds, after the two $|\beta|^2$ pieces combine into the standard coherent-state normalization $-|\beta|^2/2$
\begin{equation}
\bra\beta S(\xi)\ket\alpha=\frac{1}{\sqrt\mu}\,
\exp\!\left[-\frac{|\beta|^2}2-\frac{|\alpha|^2}2+\frac{\beta^*\alpha}\mu-\frac{\nu}{2\mu}\beta^{*2}+\frac{\nu^*}{2\mu}\alpha^2\right].
\label{eq:beta-S-alpha}
\end{equation}
Projecting onto the number state via $\braket n\beta=e^{-|\beta|^2/2}\beta^n/\sqrt{n!}$ and the resolution of the identity gives
\begin{equation}
\bra n S(\xi)\ket\alpha=\frac{1}{\pi\sqrt{n!\,\mu}}\,
\exp\!\left[-\frac{|\alpha|^2}2+\frac{\nu^*}{2\mu}\alpha^2\right]
\int d^2\beta\,\beta^n\,
\exp\!\left[-|\beta|^2-\frac{\nu}{2\mu}\beta^{*2}+\frac{\alpha}\mu\beta^*\right],
\label{MasterIdentitySqueeze}
\end{equation}
a Gaussian integral with a quadratic phase term. Its exact evaluation reproduces the known squeezed coherent-state matrix element in closed form
\begin{equation}
\bra n S(\xi)\ket\alpha=\frac{1}{\sqrt{n!\,\mu}}\left(\frac{\nu}{2\mu}\right)^{n/2}
\exp\!\left[-\frac{|\alpha|^2}2+\frac{\nu^*}{2\mu}\alpha^2\right]
H_n\!\left(\frac{\alpha}{\sqrt{2\mu\nu}}\right),
\label{eq:squeeze-coherent-ME}
\end{equation}
where $H_n$ is the standard Hermite polynomial. Since $H_n$ has the parity of $n$, the product of $(\nu/2\mu)^{n/2}$ with the Hermite polynomial in \eqref{eq:squeeze-coherent-ME} is a polynomial in $\nu/\mu$ and $\alpha$; the kernel is therefore independent of the branch chosen for the square roots. This squeezed transition kernel is denoted $K_n(\xi;\alpha)\equiv\bra n S(\xi)\ket\alpha$. This expression reduces correctly to $\bra n\ket\alpha=e^{-|\alpha|^2/2}\alpha^n/\sqrt{n!}$ as $r\to0$, and to the well-known even-photon squeezed-vacuum distribution $\bra{2m}S(\xi)\ket0=\sqrt{(2m)!}\,(-e^{i\theta}\tanh r/2)^m/(m!\sqrt{\cosh r})$ at $\alpha=0$\cite{Yuen1976,CavesSchumaker1985I,BarnettRadmore}. The analogous, more involved two-mode version is carried out in full in Sec.~\ref{sec:2mode-squeeze}.

\subsection{Two-Mode Squeezing Operator}
\label{sec:2mode-squeeze}

This subsection derives the Fock amplitudes of the two-mode squeezer acting on a two-mode coherent state. Consider
\begin{equation}
S_2(\zeta) = \exp\!\left(\zeta a^\dagger b^\dagger-\zeta^* ab\right), \qquad \zeta=re^{i\theta},
\label{S2_def}
\end{equation}
acting on the two-mode coherent state $\ket{\alpha,\lambda}=\ket\alpha\otimes\ket\lambda$, with the two-mode resolution of the identity $I=\int\frac{d^2\beta\,d^2\gamma}{\pi^2}\,\ketbra{\beta,\gamma}{\beta,\gamma}$. Write $\mu=\cosh r$, $\nu=e^{i\theta}\sinh r$, $|\mu|^2-|\nu|^2=1$. As in the single-mode case, the Bogoliubov transformation follows from the Baker--Campbell--Hausdorff expansion of $S_2^\dagger(\zeta)\,a\,S_2(\zeta)$ and $S_2^\dagger(\zeta)\,b\,S_2(\zeta)$
\begin{equation}
S_2^\dagger(\zeta)\,a\,S_2(\zeta)= \mu a+\nu b^\dagger,
\qquad
S_2^\dagger(\zeta)\,b\,S_2(\zeta)= \mu b+\nu a^\dagger,
\label{eq:bogoliubov-2mode}
\end{equation}
following the same convention as Refs.~\cite{Yuen1976,SchumakerCaves1985II}. The sign of the $\nu$ term differs from the single-mode relation of Sec.~\ref{sec:singlemode-squeeze} because the creation term enters the generator of $S(\xi)$ with a minus sign and that of $S_2(\zeta)$ with a plus sign; the difference is a convention.

\paragraph{Step 1: normal-ordered disentangling.} The standard Baker--Campbell--Hausdorff disentangling of the two-mode squeezing operator (see e.g.\ Ref.~\cite{BarnettRadmore}) reads
\begin{equation}
S_2(\zeta) = \frac{1}{\mu}\,\exp\!\left(\frac{\nu}{\mu}a^\dagger b^\dagger\right)\exp\!\left[-\ln\mu\,(a^\dagger a+b^\dagger b)\right]\exp\!\left(-\frac{\nu^*}{\mu}ab\right).
\label{eq:BCH2mode}
\end{equation}

\paragraph{Step 2: action on the coherent state.} Since $ab\ket{\alpha,\lambda}=\alpha\lambda\ket{\alpha,\lambda}$
\begin{equation}
\exp\!\left(-\frac{\nu^*}{\mu}ab\right)\ket{\alpha,\lambda} = e^{-\frac{\nu^*}{\mu}\alpha\lambda}\,\ket{\alpha,\lambda} .
\end{equation}
Using the elementary identity $e^{tN}\ket\alpha=\exp\!\left[\tfrac12|\alpha|^2(e^{2t}-1)\right]\ket{\alpha e^t}$ (immediate from comparing Fock-basis coefficients) with $N=a^\dagger a+b^\dagger b$ and $t=-\ln\mu$, so that $e^t=1/\mu$
\begin{equation}
\exp\!\left[-\ln\mu\,(a^\dagger a+b^\dagger b)\right]\ket{\alpha,\lambda}
= \exp\!\left[\tfrac12(|\alpha|^2+|\lambda|^2)\Bigl(\tfrac1{\mu^2}-1\Bigr)\right]\left|\frac{\alpha}{\mu},\frac{\lambda}{\mu}\right\rangle .
\end{equation}

\paragraph{Step 3: bra-side action of the creation exponential.} Since $ab\ket{\beta,\gamma}=\beta\gamma\ket{\beta,\gamma}$, the adjoint relation gives $\bra{\beta,\gamma}\exp\!\left(\frac{\nu}{\mu}a^\dagger b^\dagger\right)=e^{\frac{\nu}{\mu}\beta^*\gamma^*}\bra{\beta,\gamma}$.

\paragraph{Step 4: assembling the overlap kernel.} Combining Steps 1--3 and the two-mode coherent-state overlap $\braket{\beta,\gamma}{p,q}=\exp\!\left[-\tfrac12|\beta|^2-\tfrac12|p|^2+\beta^*p-\tfrac12|\gamma|^2-\tfrac12|q|^2+\gamma^*q\right]$ with $(p,q)=(\alpha/\mu,\lambda/\mu)$, and collecting the $|\alpha|^2,|\lambda|^2$ exponents [which combine as $\tfrac12(1/\mu^2-1)-\tfrac1{2\mu^2}=-\tfrac12$, restoring the expected normalization], one obtains the closed-form Gaussian kernel
\begin{equation}
\bra{\beta,\gamma}S_2(\zeta)\ket{\alpha,\lambda}
= \frac1\mu\,\exp\!\left[-\tfrac12\bigl(|\alpha|^2+|\lambda|^2+|\beta|^2+|\gamma|^2\bigr)
+\frac{\beta^*\alpha+\gamma^*\lambda}{\mu}-\frac{\nu^*}{\mu}\alpha\lambda+\frac{\nu}{\mu}\beta^*\gamma^*\right].
\label{eq:S2kernel}
\end{equation}

\paragraph{Step 5: the master identity and its exact evaluation.} By Proposition~\ref{cor:single}
\begin{equation}
\langle n,m|S_2(\zeta)|\alpha,\lambda\rangle = \int\frac{d^2\beta\,d^2\gamma}{\pi^2}\,\langle n,m|\beta,\gamma\rangle\,\bra{\beta,\gamma}S_2(\zeta)\ket{\alpha,\lambda} ,
\end{equation}
and using $\langle n,m|\beta,\gamma\rangle=e^{-\frac12|\beta|^2-\frac12|\gamma|^2}\beta^n\gamma^m/\sqrt{n!m!}$ together with \eqref{eq:S2kernel}, the $|\beta|^2,|\gamma|^2$ terms combine to $-|\beta|^2-|\gamma|^2$, leaving the two-dimensional integral
\begin{equation}
\langle n,m|S_2(\zeta)|\alpha,\lambda\rangle
= \frac{e^{-\frac12|\alpha|^2-\frac12|\lambda|^2-\frac{\nu^*}{\mu}\alpha\lambda}}{\mu\sqrt{n!m!}}
\int\frac{d^2\beta\,d^2\gamma}{\pi^2}\,\beta^n\gamma^m\,
\exp\!\left[-|\beta|^2-|\gamma|^2+\frac{\alpha}{\mu}\beta^*+\frac{\lambda}{\mu}\gamma^*+\frac{\nu}{\mu}\beta^*\gamma^*\right].
\end{equation}
Expanding $e^{(\nu/\mu)\beta^*\gamma^*}$ in powers of $\beta^*\gamma^*$ and applying \eqref{eq:MI2} to each of the resulting one-dimensional integrals over $\beta$ and over $\gamma$ separately gives
\begin{align}
\langle n,m|S_2(\zeta)|\alpha,\lambda\rangle
&= \frac{e^{-\frac12|\alpha|^2-\frac12|\lambda|^2-\frac{\nu^*}{\mu}\alpha\lambda}}{\mu\sqrt{n!m!}} \sum_{k=0}^{\min(n,m)} \frac{1}{k!} \left( \frac{\nu}{\mu} \right)^k \nonumber\\
&\quad\times\left[ \int \frac{d^2\beta}{\pi} e^{-|\beta|^2+\beta^*(\alpha/\mu)} \beta^n (\beta^*)^k \right] \left[ \int \frac{d^2\gamma}{\pi} e^{-|\gamma|^2+\gamma^*(\lambda/\mu)} \gamma^m (\gamma^*)^k \right] \nonumber \\
&= \frac{e^{-\frac12|\alpha|^2-\frac12|\lambda|^2-\frac{\nu^*}{\mu}\alpha\lambda}}{\mu\sqrt{n!m!}} \sum_{k=0}^{\min(n,m)} \frac{1}{k!} \left( \frac{\nu}{\mu} \right)^k \frac{n!}{(n-k)!} \left( \frac{\alpha}{\mu} \right)^{n-k} \frac{m!}{(m-k)!} \left( \frac{\lambda}{\mu} \right)^{m-k}.
\end{align}
Collecting the powers of $\mu$,
\begin{equation}
\mu^{-1}\,\mu^{-k}\,\mu^{-(n-k)}\,\mu^{-(m-k)}=\mu^{-(n+m+1)}\,\mu^{k},
\end{equation}
one obtains
\begin{equation}
\boxed{\;
\langle n,m|S_2(\zeta)|\alpha,\lambda\rangle
= \frac{e^{-\frac12|\alpha|^2-\frac12|\lambda|^2-\frac{\nu^*}{\mu}\alpha\lambda}}{\sqrt{n!m!}} \mu^{-(n+m+1)}
\sum_{k=0}^{\min(n,m)}\frac{n!\,m!}{k!(n-k)!(m-k)!}\,\alpha^{n-k}\lambda^{m-k} \left(\mu\nu\right)^k .
\;}
\label{eq:S2-final}
\end{equation}
\emph{Consistency check.} At $\alpha=\lambda=0$ (vacuum input), only the $k=n=m$ term survives (since $0^0=1$). In this limit, the sum collapses to a single term with $k=n$, yielding
\begin{equation}
\langle n,n|S_2(\zeta)|0,0\rangle = \frac{1}{n!} \mu^{-(2n+1)} \frac{(n!)^2}{n!} (\mu\nu)^n = \frac{1}{\mu} \left( \frac{\nu}{\mu} \right)^n,
\end{equation}
which is exactly the standard two-mode squeezed-vacuum photon-number amplitude\cite{Yuen1976,SchumakerCaves1985II}. This confirms the correctness of the normalization carried through Steps 1--5.

\subsection{Beam Splitter Operator}
\label{sec:beamsplitter}

This subsection treats a passive two-mode operation, the beam splitter $U_{BS}(\varphi)$. The underlying operator identity is
\[
U_{BS}(\varphi)\ket{\alpha,\lambda}
=
U_{BS}(\varphi)\,\mathbb{I}\ket{\alpha,\lambda},
\]
where the identity operator is resolved in the two-mode coherent-state basis. The single-insertion identity of Proposition~\ref{cor:single} gives
\begin{equation}
\langle n,m|U_{BS}(\varphi)|\alpha,\lambda\rangle
=
\int
\frac{d^2\beta\,d^2\gamma}{\pi^2}\,
\langle n,m|\beta,\gamma\rangle
\langle\beta,\gamma|U_{BS}(\varphi)|\alpha,\lambda\rangle .
\end{equation}
To evaluate the kernel $\bra{\beta,\gamma}U_{BS}(\varphi)\ket{\alpha,\lambda}$, the explicit form is used
\begin{equation}
U_{BS}(\varphi)
=
\exp\!\left[\varphi\left(a^\dagger b-a b^\dagger\right)\right],
\qquad
\varphi\in\mathbb R ,
\label{BS_def}
\end{equation}
together with the Bogoliubov transformation
\[
U_{BS}^\dagger a\,U_{BS}=ca+sb,\qquad
U_{BS}^\dagger b\,U_{BS}=cb-sa,
\]
where
\[
c=\cos\varphi,\qquad
s=\sin\varphi,\qquad
c^2+s^2=1,
\]
which give the coherent-state map\cite{BarnettRadmore}
\[
U_{BS}(\varphi)\ket{\alpha,\lambda}
=
\ket{c\alpha+s\lambda,\,
c\lambda-s\alpha},
\]
with unit phase factor, since $U_{BS}$ conserves the total photon number and leaves the vacuum invariant. Hence
\[
\langle\beta,\gamma|U_{BS}(\varphi)|\alpha,\lambda\rangle
=
\langle\beta,\gamma|
c\alpha+s\lambda,\,
c\lambda-s\alpha
\rangle,
\]
and substituting into the projection integral gives
\begin{equation}
\langle n,m|U_{BS}(\varphi)|\alpha,\lambda\rangle
=
\frac{e^{-\frac12(|\alpha|^2+|\lambda|^2)}}{\pi^2\sqrt{n!m!}}
\int
d^2\beta\,d^2\gamma\,
\beta^n\gamma^m
\exp\!\left[
-|\beta|^2-|\gamma|^2
+\beta^*(c\alpha+s\lambda)
+\gamma^*(c\lambda-s\alpha)
\right].
\label{eq:BS-master}
\end{equation}
The kernel factorizes into two Gaussian integrals, each evaluated by Eq.~\eqref{eq:MI2}, and
\begin{equation}
\boxed{
\;
\langle n,m|U_{BS}(\varphi)|\alpha,\lambda\rangle
=
e^{-\frac12(|\alpha|^2+|\lambda|^2)}
\frac{(c\alpha+s\lambda)^n}{\sqrt{n!}}
\frac{(c\lambda-s\alpha)^m}{\sqrt{m!}} .
\;
}
\label{BS_explicit}
\end{equation}

The two-mode squeezer couples the two integration variables and produces a finite sum, whereas the passive beam splitter maps coherent states into coherent states and the integral factorizes.

\subsection{A Compound Two-Mode Network: Beam Splitter after Two-Mode Squeezing}
\label{sec:composite}

This subsection combines the two previous operations into an exact Fock-basis identity for a compound Gaussian network built from the two operators of the previous two subsections
\begin{equation}
W = U_{BS}(\varphi)\,S_2(\zeta) ,
\end{equation}
i.e.\ a two-mode squeezer followed by a beam splitter---exactly the type of building block used in continuous-variable photonic quantum circuits and in Gaussian boson sampling architectures such as Xanadu's Borealis processor \cite{Madsen2022}.

The derivation uses the fact that $U_{BS}(\varphi)$ acts on coherent-state \emph{bras} exactly as it acts on kets, with $\varphi\to-\varphi$: since $U_{BS}(\varphi)^\dagger=U_{BS}(-\varphi)$ and $U_{BS}(-\varphi)\ket{\beta,\gamma}=\ket{c\beta-s\gamma,\,c\gamma+s\beta}$ (the coherent-state map of Sec.~\ref{sec:beamsplitter} evaluated at $\varphi\to-\varphi$), it follows that $\bra{\beta,\gamma}U_{BS}(\varphi) = \bra{c\beta-s\gamma,\,c\gamma+s\beta}$. Inserting the two-mode resolution of the identity via the master identity (equivalently, a single application of Proposition~\ref{cor:single} to $\mathcal F=W$) gives
\begin{equation}
\langle n,m|W|\alpha,\lambda\rangle
= \int\frac{d^2\beta\,d^2\gamma}{\pi^2}\,\langle n,m|\beta,\gamma\rangle\;
\bra{c\beta-s\gamma,\;c\gamma+s\beta}\,S_2(\zeta)\,\ket{\alpha,\lambda} ,
\label{eq:composite-master}
\end{equation}
where the second factor is precisely the kernel \eqref{eq:S2kernel} evaluated at the rotated arguments $\beta'=c\beta-s\gamma$, $\gamma'=c\gamma+s\beta$ (a measure-preserving rotation, so $d^2\beta\,d^2\gamma=d^2\beta'\,d^2\gamma'$ and $|\beta'|^2+|\gamma'|^2=|\beta|^2+|\gamma|^2$). Equation \eqref{eq:composite-master} is an exact integral identity for the compound network, obtained by substituting the kernel \eqref{eq:S2kernel} into the bra transformation of the beam splitter; it illustrates the compositional character of the insertion.

\paragraph{Completing the remaining Gaussian integral.}
Substituting $\beta'=c\beta-s\gamma$, $\gamma'=c\gamma+s\beta$ into the exponent of \eqref{eq:S2kernel} and collecting terms gives, after the $|\beta|^2,|\gamma|^2$ pieces combine with those from $\langle n,m|\beta,\gamma\rangle$ exactly as in Sec.~\ref{sec:2mode-squeeze}
\begin{equation}
\langle n,m|W|\alpha,\lambda\rangle
= \frac{e^{-\frac12|\alpha|^2-\frac12|\lambda|^2-\frac{\nu^*}{\mu}\alpha\lambda}}{\mu\sqrt{n!m!}}
\int\frac{d^2\beta\,d^2\gamma}{\pi^2}\,\beta^n\gamma^m\,e^{\mathcal E(\beta,\gamma)} ,
\label{eq:composite-exponent}
\end{equation}
with the exponent
\begin{equation}
\mathcal E(\beta,\gamma)=-|\beta|^2-|\gamma|^2+A\beta^*+\Lambda\gamma^*+\tfrac{b_1}{2}(\beta^*)^2+\tfrac{b_2}{2}(\gamma^*)^2+h\,\beta^*\gamma^* ,
\end{equation}
where the rotation of the linear terms produces \emph{effective displacements}
\begin{equation}
A = \frac{c\alpha+s\lambda}{\mu}, \qquad \Lambda = \frac{c\lambda-s\alpha}{\mu} ,
\end{equation}
and the rotation of the cross term $(\nu/\mu)\beta'^*\gamma'^*$ produces \emph{same-mode} quadratic terms in addition to the cross term. Explicitly, with $\beta'=c\beta-s\gamma$ and $\gamma'=c\gamma+s\beta$,
\begin{equation}
\beta'^*\gamma'^* = (c^2-s^2)\,\beta^*\gamma^* + cs\bigl[(\beta^*)^2-(\gamma^*)^2\bigr]
= \cos(2\varphi)\,\beta^*\gamma^* + \tfrac12\sin(2\varphi)\bigl[(\beta^*)^2-(\gamma^*)^2\bigr] ,
\label{eq:rotated-cross}
\end{equation}
so that the coefficients appearing in \eqref{eq:composite-exponent} are
\begin{equation}
b_1 = \frac{\nu}{\mu}\sin(2\varphi), \qquad b_2 = -\frac{\nu}{\mu}\sin(2\varphi), \qquad h = \frac{\nu}{\mu}\cos(2\varphi) .
\label{eq:b1b2h}
\end{equation}
This is the key structural point: because the beam-splitter rotation mixes $\beta^*$ and $\gamma^*$ \emph{before} they are multiplied together in the cross term of \eqref{eq:S2kernel}, the composite kernel acquires self-squeezing terms $(\beta^*)^2$ and $(\gamma^*)^2$ in each mode, in addition to the cross term $\beta^*\gamma^*$ of $S_2(\zeta)$.

The integral \eqref{eq:composite-exponent} is evaluated exactly by the method of \eqref{eq:MI2}--\eqref{eq:S2-final}. The quadratic part of the exponent is expanded as a triple power series,
\begin{equation}
e^{\frac{b_1}{2}(\beta^*)^2+h\beta^*\gamma^*+\frac{b_2}{2}(\gamma^*)^2}
=\sum_{p,q,k\ge0}\frac{(b_1/2)^p\,h^q\,(b_2/2)^k}{p!\,q!\,k!}\,(\beta^*)^{2p+q}(\gamma^*)^{2k+q},
\end{equation}
and each monomial is integrated against $\beta^n\gamma^m e^{-|\beta|^2-|\gamma|^2+A\beta^*+\Lambda\gamma^*}$ with \eqref{eq:MI2} in each variable separately. With $|\nu/\mu|=\tanh r$ and $|\sin2\varphi|+|\cos2\varphi|\le\sqrt2$, the moduli of the expanded terms sum to at most
\begin{equation}
\exp\!\Bigl[\tfrac{|b_1|}{2}|\beta|^2+\tfrac{|b_2|}{2}|\gamma|^2+|h|\,|\beta||\gamma|\Bigr]
\le\exp\!\Bigl[\tfrac{\tanh r}{\sqrt2}\bigl(|\beta|^2+|\gamma|^2\bigr)\Bigr],
\end{equation}
and since $\tanh r/\sqrt2<1$ the Gaussian factor $e^{-|\beta|^2-|\gamma|^2}$ dominates, which justifies the term-by-term integration. This gives the exact, closed double-index sum
\begin{equation}
\langle n,m|W|\alpha,\lambda\rangle
= \frac{e^{-\frac12|\alpha|^2-\frac12|\lambda|^2-\frac{\nu^*}{\mu}\alpha\lambda}}{\mu\sqrt{n!m!}}\;
\mathcal H_{n,m}(A,\Lambda;b_1,b_2,h) ,
\label{eq:composite-final-corrected}
\end{equation}
where $\mathcal H_{n,m}$ is the two-mode (bivariate) Hermite-type polynomial
\begin{equation}
\boxed{\;
\begin{aligned}
\mathcal H_{n,m}(A,\Lambda;b_1,b_2,h)
:= n!\,m!\sum_{q=0}^{\min(n,m)}\frac{h^{q}}{q!}
&\sum_{p=0}^{\lfloor (n-q)/2\rfloor}\frac{(b_1/2)^{p}}{p!\,(n-2p-q)!}\,A^{n-2p-q}\\
\times&\sum_{k=0}^{\lfloor (m-q)/2\rfloor}\frac{(b_2/2)^{k}}{k!\,(m-2k-q)!}\,\Lambda^{m-2k-q} .
\end{aligned}\;}
\label{eq:bivariate-hermite}
\end{equation}
Polynomials of this general type---generating functions of the form $\exp(xt+yt^2)$ and their two-mode products linked by a cross-coupling parameter---are closely related to those appearing in the loop-hafnian formulation of Gaussian boson sampling amplitudes \cite{Madsen2022,Hamilton2017,Kruse2019,Bjorklund2019,Quesada2019}; \eqref{eq:composite-final-corrected}--\eqref{eq:bivariate-hermite} is the exact two-mode special case of that general multimode formula, obtained here from the kernels of the two factors.

\emph{Consistency checks.} At $\varphi=0$ ($c=1,s=0$): $b_1=b_2=0$, $h=\nu/\mu$, $A=\alpha/\mu$, $\Lambda=\lambda/\mu$, and only the $p=k=0$ terms survive, so $\mathcal H_{n,m}\to n!m!\sum_q (h^q/q!)\,A^{n-q}\Lambda^{m-q}/[(n-q)!(m-q)!]$, which reproduces \eqref{eq:S2-final} exactly. At $r=0$ ($\mu=1,\nu=0$): $b_1=b_2=h=0$, so only $p=k=q=0$ survives, $\mathcal H_{n,m}=n!m!A^n\Lambda^m$, reproducing \eqref{BS_explicit}. 

\paragraph{Other routes.} The same amplitudes follow from the symplectic composition of the two Gaussian operations followed by a loop-hafnian evaluation \cite{Hamilton2017,Bjorklund2019,Quesada2019}, or from Lie disentangling of the composite generator. The correspondence is explicit. Let
\begin{equation}
B=\begin{pmatrix} b_1 & h\\ h & b_2\end{pmatrix},
\end{equation}
and let $\tilde B_{(n,m)}$ be the $(n+m)\times(n+m)$ matrix obtained by repeating the first row and column of $B$ $n$ times and the second $m$ times, with the diagonal entries replaced by $A$ for the first $n$ indices and by $\Lambda$ for the last $m$. Then
\begin{equation}
\mathcal H_{n,m}(A,\Lambda;b_1,b_2,h)=\operatorname{lhaf}\bigl(\tilde B_{(n,m)}\bigr),
\label{eq:lhaf}
\end{equation}
since $n!\,m!$ times the coefficient of $z_1^nz_2^m$ in $\exp(\tfrac12 z^{\mathsf T}Bz+Az_1+\Lambda z_2)$ is a loop hafnian of this form \cite{Bjorklund2019,Quesada2019}; a direct comparison for $n,m\le5$ agrees to $8\times10^{-16}$. The derivation above uses only the kernels of the two factors and one Gaussian integral.

\subsection{A Non-Gaussian Example: Kerr-Squeezed Coherent-State Overlap}
\label{sec:kerr}

This subsection applies the nested insertion to a non-Gaussian matrix element: the overlap generated by a Kerr (self-phase-modulation) nonlinearity followed by squeezing
\begin{equation}
\mathcal F = S(\xi)\,e^{i\chi N^2},
\qquad
S(\xi)=\exp\Bigl[\tfrac12(\xi^*a^2-\xi a^{\dagger2})\Bigr],
\qquad
N=a^\dagger a .
\end{equation}
Here $\chi\in\mathbb R$ is the Kerr phase. This composite operator models the sequence ``Kerr phase evolution, then squeezing'' familiar from Kerr-nonlinear resonators, in which the single-photon Kerr effect is known to produce fractional revivals of an initially coherent wavepacket, as observed experimentally by Kirchmair \textit{et al.}\ \cite{Kirchmair2013} and further engineered in driven Kerr-nonlinear resonators for cat-state preparation \cite{PuriBoutinBlais2017}. Because $e^{i\chi N^2}$ introduces a quadratic phase in the Fock index, $\bra\beta e^{i\chi N^2}\ket\alpha$ is a non-Gaussian series of Jacobi-theta type.

Applying Theorem~\ref{thm:telescope} with $\mathcal F_1=e^{i\chi N^2}$ and $\mathcal F_2=S(\xi)$, the resolution of the identity is inserted \emph{between} the Kerr phase and the squeezing operator
\begin{equation}
\langle n|S(\xi)\,e^{i\chi N^2}|\alpha\rangle
= \int\frac{d^2\beta}{\pi}\,\langle n|S(\xi)|\beta\rangle\,\langle\beta|e^{i\chi N^2}|\alpha\rangle .
\label{eq:kerr-telescope}
\end{equation}
The second factor is computed directly from the Fock expansion of $\ket\alpha$
\begin{equation}
\langle\beta|e^{i\chi N^2}|\alpha\rangle
= e^{-\frac12|\alpha|^2-\frac12|\beta|^2}\sum_{k=0}^\infty \frac{(\alpha\beta^*)^k}{k!}\,e^{i\chi k^2},
\label{eq:kerr-symbol}
\end{equation}
The series converges absolutely, and termwise estimation gives the Gaussian bound
\begin{equation}
\bigl|\langle\beta|e^{i\chi N^2}|\alpha\rangle\bigr|
\le e^{-\frac12|\alpha|^2-\frac12|\beta|^2}\sum_{k=0}^\infty\frac{|\alpha|^k|\beta|^k}{k!}
=e^{-\frac12(|\beta|-|\alpha|)^2},
\label{eq:kerr-bound}
\end{equation}
which is integrable over the plane. Hence Assumption~4 holds for the vector $e^{i\chi N^2}\ket\alpha$, and since $S(\xi)$ is unitary, $\|S(\xi)\ket\beta\|=1$ and Assumption~5 holds for the pair $(S(\xi),e^{i\chi N^2}\ket\alpha)$. The first factor $\bra n S(\xi)\ket\beta$ is the squeezed coherent-state transition kernel derived in closed form in Sec.~\ref{sec:singlemode-squeeze}, Eq.~\eqref{eq:squeeze-coherent-ME}, which we denote here as $K_n(\xi;\beta)$. Substituting \eqref{eq:kerr-symbol} into \eqref{eq:kerr-telescope} gives the exact integral identity
\begin{equation}
\boxed{\;
\langle n|S(\xi)e^{i\chi N^2}|\alpha\rangle
= \int\frac{d^2\beta}{\pi}\, K_n(\xi;\beta)\; e^{-\frac12|\alpha|^2-\frac12|\beta|^2}\sum_{k=0}^\infty \frac{(\alpha\beta^*)^k}{k!}\,e^{i\chi k^2} .
\;}
\label{eq:kerr-result}
\end{equation}
Equation \eqref{eq:kerr-result} converts the double nonlinearity (Kerr phase composed with squeezing) into a single well-defined two-dimensional phase-space integral against a known Gaussian--Hermite kernel.

\paragraph{Rational Kerr phases.}
For rational $\chi/\pi$ the kernel collapses to a finite sum.

\begin{proposition}
\label{prop:kerr}
Let $\chi=\pi p/q$ with integers $p$ and $q\ge1$, let $L=2q$ and $\omega=e^{2\pi i/L}$, and define
\begin{equation}
g_j:=\frac1L\sum_{k=0}^{L-1}e^{i\chi k^2}\,\omega^{-jk},\qquad j=0,\dots,L-1 .
\label{eq:gauss}
\end{equation}
\begin{enumerate}
\item[(i)] $e^{i\chi N^2}=\sum_{j=0}^{L-1}g_j\,\omega^{j N}$; in particular $e^{i\chi N^2}\ket\alpha=\sum_{j=0}^{L-1}g_j\ket{\alpha\omega^j}$.
\item[(ii)] The matrix element is
\begin{equation}
\bra nS(\xi)e^{i\chi N^2}\ket\alpha=\sum_{j=0}^{L-1}g_j\,K_n(\xi;\alpha\omega^j),
\label{eq:kerrclosed}
\end{equation}
with $K_n$ given by \eqref{eq:squeeze-coherent-ME}.
\item[(iii)] If $p$ and $q$ are coprime, then $g_j\neq0$ exactly when $j\equiv pq\pmod 2$, and in that case $|g_j|=q^{-1/2}$. Exactly $q$ of the $L$ coefficients are non-zero.
\end{enumerate}
\end{proposition}

\begin{proof}
(i) Since $\chi[(k+L)^2-k^2]=4\pi p(k+q)\in2\pi\mathbb Z$, the sequence $e^{i\chi k^2}$ has period $L$, and discrete Fourier inversion gives $e^{i\chi k^2}=\sum_jg_j\omega^{jk}$ for all $k\ge0$. Both sides of the operator identity are bounded and agree on the number basis. Since $\omega^{j N}\ket\alpha=\ket{\alpha\omega^j}$ by the Fock expansion of $\ket\alpha$, the second statement follows.

(ii) By (i), $\bra\beta e^{i\chi N^2}\ket\alpha=\sum_jg_j\braket\beta{\alpha\omega^j}$. Inserting this into \eqref{eq:kerr-telescope} and applying the resolution of the identity \eqref{ROI} to each term gives \eqref{eq:kerrclosed}.

(iii) The summand in \eqref{eq:gauss} is $f_j(k)=e^{i\pi(pk^2-jk)/q}$, and $f_j(k+q)=e^{i\pi(pq-j)}f_j(k)$. The terms $k$ and $k+q$ in \eqref{eq:gauss} therefore cancel for $j\not\equiv pq\pmod2$. In the remaining case $f_j$ has period $q$, and $g_j=S_j/q$ with $S_j=\sum_{k=0}^{q-1}f_j(k)$. Using the periodicity and $f_j(k+h)\overline{f_j(k)}=f_j(h)\,e^{2\pi iphk/q}$,
\[
|S_j|^2=\sum_{h=0}^{q-1}f_j(h)\sum_{k=0}^{q-1}e^{2\pi iphk/q}=q ,
\]
because for coprime $p$ and $q$ the inner sum equals $q\,\delta_{h,0}$. Hence $|g_j|=q^{-1/2}$. Exactly $q$ of the indices $j=0,\dots,2q-1$ have the parity of $pq$.
\end{proof}

The coefficients \eqref{eq:gauss} are normalised quadratic Gauss sums \cite{Spiridonov2022}. Statement (i) is the fractional-revival decomposition of Kerr dynamics into finitely many rotated coherent states \cite{YurkeStoler1986,TaraAgarwalChaturvedi1993,WuGuo1996}, observed experimentally in Ref.~\cite{Kirchmair2013}; in operator form it writes the Kerr phase at rational times as a finite combination of phase rotations, the form used for Kerr gates in Gaussian circuits in Ref.~\cite{Upreti2026}. Statement (iii) makes the count explicit: for $p/q$ in lowest terms, exactly $q$ rotated components occur, all with the same weight $|g_j|^2=1/q$. Equation~\eqref{eq:kerrclosed} then expresses the Fock amplitudes of the squeezed revival states in closed form, as Hermite polynomials evaluated at the rotated amplitudes $\alpha\omega^j$.

The integral \eqref{eq:kerr-result} converges absolutely for every real $\chi$, by \eqref{eq:kerr-bound} and the integrability of $K_n(\xi;\beta)$, which follows from $|\nu/(2\mu)|=\frac12\tanh r<\frac12$. It is therefore exact also for irrational $\chi/\pi$, where the kernel is an infinite series, and the finite sums \eqref{eq:kerrclosed} approximate it with an explicit error.

\begin{proposition}[Rational approximation]
\label{prop:kerrapprox}
Let $\chi,\chi'\in\mathbb R$ and $x=|\alpha|^2$. For every $n\ge0$
\[
\bigl|\bra nS(\xi)e^{i\chi N^2}\ket\alpha-\bra nS(\xi)e^{i\chi' N^2}\ket\alpha\bigr|\le|\chi-\chi'|\,\sqrt{x^4+6x^3+7x^2+x}\,.
\]
For every $Q\ge1$ there are coprime integers $p$ and $q$ with $q\le Q$ and $|\chi/\pi-p/q|\le1/(qQ)$. For $\chi'=\pi p/q$ the sum \eqref{eq:kerrclosed} has $q$ non-zero terms and approximates the matrix element at $\chi$ with error at most $\pi\sqrt{x^4+6x^3+7x^2+x}/q^2$.
\end{proposition}

\begin{proof}
Since $S(\xi)$ is unitary and $|\braket n\phi|\le\|\phi\|$, the left side is at most $\|(e^{i\chi N^2}-e^{i\chi' N^2})\ket\alpha\|$. Using $|e^{i\chi k^2}-e^{i\chi'k^2}|\le|\chi-\chi'|k^2$ we have
\[
\bigl\|(e^{i\chi N^2}-e^{i\chi' N^2})\ket\alpha\bigr\|^2
=e^{-x}\sum_{k\ge0}\frac{x^k}{k!}\,|e^{i\chi k^2}-e^{i\chi'k^2}|^2
\le|\chi-\chi'|^2\,e^{-x}\sum_{k\ge0}\frac{x^k}{k!}\,k^4 ,
\]
and the last sum equals $x^4+6x^3+7x^2+x$. For the second statement, Dirichlet's approximation theorem gives integers $p$ and $q$ with $1\le q\le Q$ and $|q\chi/\pi-p|<1/Q$. Cancelling a common factor preserves $|\chi/\pi-p/q|<1/(qQ)$ and keeps $q\le Q$. Then $|\chi-\chi'|=\pi|\chi/\pi-p/q|\le\pi/(qQ)\le\pi/q^2$, and Proposition~\ref{prop:kerr}(iii) gives the number of non-zero terms.
\end{proof}

The proof bounds the full vector $(e^{i\chi N^2}-e^{i\chi' N^2})\ket\alpha$, so the same estimate holds with $S(\xi)$ replaced by any unitary. Along the convergents of the continued fraction of $\chi/\pi$, Hurwitz's theorem \cite{Hurwitz1891} gives
\begin{equation}
\Bigl|\frac{\chi}{\pi}-\frac pq\Bigr|<\frac{1}{\sqrt5\,q^2}
\end{equation}
for infinitely many $q$, which improves the error bound of Proposition~\ref{prop:kerrapprox} by the factor $1/\sqrt5$ for those $q$; the same Diophantine route controls Kerr gates in Gaussian circuits in Ref.~\cite{Upreti2026}.

\paragraph{Computational route.} At irrational $\chi/\pi$ the most efficient numerical route is the direct Fock sum
\begin{equation}
\bra n S(\xi)e^{i\chi N^2}\ket\alpha=\sum_{m=0}^{\infty}\bra n S(\xi)\ket m\,e^{i\chi m^2}\braket m\alpha ,
\label{eq:kerr-fock}
\end{equation}
whose truncation at $m<m_{\max}$ converges with the Poisson tail of $\ket\alpha$, faster than any power of $1/m_{\max}$. For $\xi=0.4$, $|\alpha|^2=0.65$ and $\chi/\pi=\sqrt2/10$, the truncation $m_{\max}=12$ reproduces the amplitudes with $n\le15$ to an accuracy of $10^{-6}$, whereas the convergent $p/q=197/1393$ of $\chi/\pi$ gives a maximal error of $1.5\times10^{-7}$ with $q=1393$ terms. For this $q$ the bound of Proposition~\ref{prop:kerrapprox} equals $3.8\times10^{-6}$ and the intermediate bound $|\chi-\chi'|\langle N^4\rangle^{1/2}$ equals $2.7\times10^{-7}$. Proposition~\ref{prop:kerrapprox} therefore serves as a certified analytic error bound, and the finite sums \eqref{eq:kerrclosed} are the efficient route at rational phases with small $q$, where the number of terms is independent of $|\alpha|$: at $\chi=\pi/3$, $\xi=0.4$ and $|\alpha|^2=100$, three terms reproduce the amplitudes with $60\le n<140$, for which a direct Fock computation needs $m_{\max}=210$ for a relative accuracy of $10^{-10}$. Since a term of the finite sum and a term of the Fock sum cost about the same, the finite sum is the efficient route when
\begin{equation}
q<m_{\max}(\varepsilon)\approx|\alpha|^2+z_\varepsilon|\alpha| ,
\label{eq:kerr-criterion}
\end{equation}
where $z_\varepsilon$ grows with the target accuracy $\varepsilon$; for $\xi=0.4$ and a relative accuracy of $10^{-10}$, the values $m_{\max}=84$, $210$ and $648$ at $|\alpha|^2=25$, $100$ and $400$ give $z_\varepsilon\approx12$.

Figure~\ref{fig:kerr} shows the distributions \eqref{eq:kerrclosed} for $q=3,4,5$ and the error as a function of the number of terms. For $q=2$ the distribution coincides with that of $\chi=0$: $K_n(\xi;-\beta)=(-1)^nK_n(\xi;\beta)$, and the two weights $g_0,g_2$ have equal modulus and a relative phase $\pi/2$.
The same finite sums hold when the squeezer is replaced by a displacement $D(\beta)$. Since $D(\beta)\ket\gamma=e^{(\beta\gamma^*-\beta^*\gamma)/2}\ket{\gamma+\beta}$, Proposition~\ref{prop:kerr}(i) gives
\begin{equation}
\bra nD(\beta)e^{i\chi N^2}\ket\alpha=\sum_{j=0}^{L-1}g_j\,e^{(\beta\gamma_j^*-\beta^*\gamma_j)/2}\,\braket n{\gamma_j+\beta},\qquad\gamma_j=\alpha\omega^j .
\label{eq:kerr-disp}
\end{equation}

\begin{corollary}[Kerr phase embedded in a network]
\label{cor:kerrnet}
Let $\chi=\pi p/q$ with $p,q$ coprime, and let $g_j$, $\omega$ and $L$ be as in Proposition~\ref{prop:kerr}. For any bounded operators $U_1,U_2$ on a space of $M$ modes, with $N_1$ the number operator of the first mode,
\[
U_2\,e^{i\chi N_1^2}\,U_1=\sum_jg_j\,U_2\,e^{2\pi ijN_1/L}\,U_1 ,
\]
where the sum runs over the $q$ indices with $g_j\neq0$. If $U_1$ and $U_2$ are Gaussian, so is every term. In particular, for the network $W=U_{BS}(\varphi)S_2(\zeta)$ of Sec.~\ref{sec:composite},
\begin{equation}
\bra{n,m}W\,e^{i\chi N_1^2}\ket{\alpha,\lambda}=\sum_jg_j\,\bra{n,m}W\ket{\alpha\omega^j,\lambda} ,
\label{eq:kerrnet}
\end{equation}
and each term is given by \eqref{eq:composite-final-corrected} with $\alpha$ replaced by $\alpha\omega^j$.
\end{corollary}

\begin{proof}
The first statement is Proposition~\ref{prop:kerr}(i) applied to the first mode, since both sides are diagonal in the number basis. For \eqref{eq:kerrnet}, $e^{2\pi ijN_1/L}\ket{\alpha,\lambda}=\ket{\alpha\omega^j,\lambda}$.
\end{proof}

For $q=2$ ($\chi=\pi/2$) we have $\omega=i$ and the two non-zero coefficients are $g_0=e^{i\pi/4}/\sqrt2$ and $g_2=e^{-i\pi/4}/\sqrt2$, so that
\[
\bra{n,m}W\,e^{i\pi N_1^2/2}\ket{\alpha,\lambda}=\frac1{\sqrt2}\Bigl(e^{i\pi/4}\bra{n,m}W\ket{\alpha,\lambda}+e^{-i\pi/4}\bra{n,m}W\ket{-\alpha,\lambda}\Bigr).
\]
For $q=3$ ($\chi=\pi/3$) we have $\omega=e^{i\pi/3}$, the non-zero coefficients have $j=1,3,5$ and modulus $3^{-1/2}$, and the three amplitudes are $\alpha e^{i\pi/3}$, $-\alpha$ and $\alpha e^{-i\pi/3}$. This set is the cube roots of unity rotated by $\pi$. Since $pq$ is odd here, the rotation by $\pi/q$ in $\omega^j$ is essential.

\paragraph{Kerr phase followed by a cubic phase gate.} Corollary~\ref{cor:kerrnet} holds for every bounded $U_2$, including the cubic phase gate $V(\lambda)=e^{i\lambda\hat x^3}$ with $\hat x=(a+a^\dagger)/\sqrt2$ and $\lambda>0$, which is diagonal in the position basis. With $U_1=I$ and $U_2=V(\lambda)$,
\begin{equation}
\bra n V(\lambda)\,e^{i\chi N^2}\ket\alpha=\sum_{j}g_j\,\bra n V(\lambda)\ket{\alpha\omega^j},
\label{eq:kerr-cubic}
\end{equation}
and each term is the Fock amplitude of the cubic phase gate on a coherent state. This amplitude has a closed form as a combination of the Airy function and its derivative at one argument shared by all Fock indices \cite{Asadi2026cubic}; the coherent input $\ket{\alpha\omega^j}$ corresponds there to the quadratic coefficient $1$ and the linear coefficient $\sqrt2\,\alpha\omega^j$ of the Gaussian exponent. At $\chi=\pi/3$ and $\lambda=0.2357$, the three terms of \eqref{eq:kerr-cubic} reproduce the amplitudes with $n\le12$, computed by position-space quadrature, to $1.5\times10^{-12}$ for $|\alpha|^2=1.44$ and to $6.6\times10^{-13}$ for $|\alpha|^2=9$, with two Airy evaluations per term for the whole Fock profile. The Fock route $\sum_m\bra n V(\lambda)\ket m\,e^{i\chi m^2}\braket m\alpha$ reaches an accuracy of $10^{-10}$ with $23$ and $47$ matrix elements per amplitude, respectively. At irrational $\chi/\pi$ the estimate of Proposition~\ref{prop:kerrapprox} applies unchanged, since $V(\lambda)$ is unitary.

\begin{figure}[t]
\centering
\includegraphics[width=\textwidth]{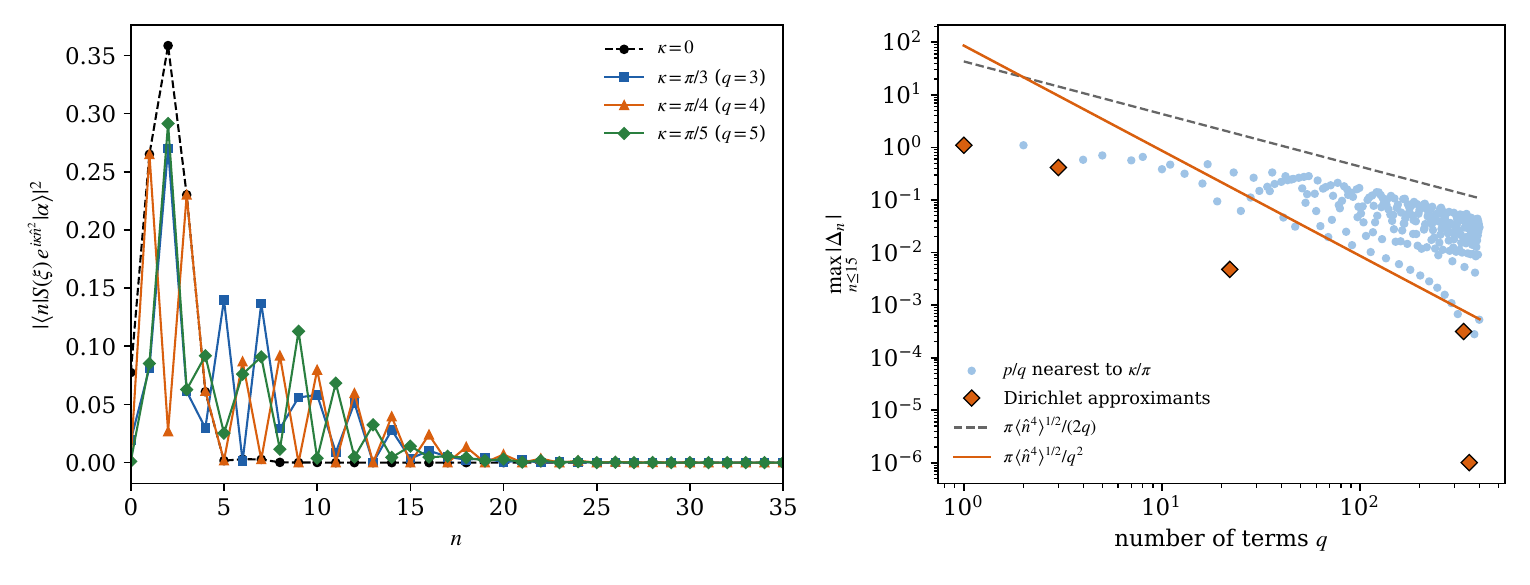}
\caption{Left: photon-number distribution $|\bra nS(\xi)e^{i\chi N^2}\ket\alpha|^2$ from \eqref{eq:kerrclosed} for $\chi=\pi/q$ with $q=3,4,5$, and for $\chi=0$ (squeezed coherent state); $\xi=0.4$, $\alpha=2$. Right: error $\Delta_n$ of the finite sum \eqref{eq:kerrclosed} at $\chi'=\pi p/q$ for $\chi=1$, maximised over $n\le15$, against the number of terms $q$. Light dots: $p/q$ nearest to $\chi/\pi$, bounded by $\pi\langle N^4\rangle^{1/2}/(2q)$ with $\langle N^4\rangle=\bra\alpha N^4\ket\alpha$. Diamonds: the approximants of Proposition~\ref{prop:kerrapprox}, bounded by $\pi\langle N^4\rangle^{1/2}/q^2$.}
\label{fig:kerr}
\end{figure}

\paragraph{Other routes.} The Kerr generator is quartic in $a$ and $a^\dagger$. After the finite decomposition of Proposition~\ref{prop:kerr}(i), each term is a phase rotation followed by a Gaussian unitary, so ordering techniques, Lie disentangling and the symplectic formalism apply term by term. This decomposition, extended to cross-Kerr gates and combined with Diophantine approximation, yields classical simulation algorithms for Gaussian circuits with Kerr gates \cite{Upreti2026}, and finite superpositions of coherent states give a complementary framework for simulating bosonic circuits \cite{MarshallAnand2023,Guseynov2026}. At irrational $\chi/\pi$ the Kerr-evolved coherent state is a one-dimensional continuous superposition of coherent states \cite{WuGuo1996}, and Eq.~\eqref{eq:kerr-result} gives the corresponding exact phase-space representation of the squeezed amplitudes.

\subsection{A Heralded Cat-State Circuit: Two Beam Splitters after a Squeezer}
\label{sec:podoshvedov}

This subsection treats a chain that requires two insertions. The circuit of Podoshvedov, Podoshvedov and Kulik \cite{Podoshvedov2023,Podoshvedov2022} provides such a chain. A single-mode squeezed vacuum is sent through a sequence of beam splitters whose auxiliary ports start in vacuum, and photon-number-resolving detection of the auxiliary modes heralds even and odd Schr\"odinger cat states in the signal mode. For two beam splitters the operator chain is
\begin{equation}
W = U^{(02)}_{BS}(\varphi_2)\,U^{(01)}_{BS}(\varphi_1)\,S_0(s) ,
\qquad t_i=\cos\varphi_i ,
\label{eq:pod-chain}
\end{equation}
acting on $\ket{0,0,0}$ and sketched in figure~\ref{fig:podoshvedov}, where $S_0(s)$ squeezes mode $0$ and $U^{(0i)}_{BS}$ mixes modes $0$ and $i$ as in Sec.~\ref{sec:beamsplitter}. With $\mathcal F_1=S_0(s)$, $\mathcal F_2=U^{(01)}_{BS}$ and $\mathcal F_3=U^{(02)}_{BS}$, Theorem~\ref{thm:telescope} applies with $N=3$ and produces the two-fold identity
\begin{equation}
\bra{n,h_1,h_2}W\ket{0,0,0}
=\int \frac{d^6\gamma_2}{\pi^3}\int \frac{d^6\gamma_1}{\pi^3}\;
\bra{n,h_1,h_2}U^{(02)}_{BS}\ket{\gamma_2}\,
\bra{\gamma_2}U^{(01)}_{BS}\ket{\gamma_1}\,
\bra{\gamma_1}S_0(s)\ket{0,0,0} ,
\label{eq:pod-nested}
\end{equation}
where $\gamma_1,\gamma_2\in\mathbb C^3$ are three-mode coherent labels and $d^6\gamma$ is the Lebesgue measure on $\mathbb C^3$. All three kernels are already available: the squeezing kernel of Sec.~\ref{sec:singlemode-squeeze} and, for each beam splitter, the coherent-state map of Sec.~\ref{sec:beamsplitter}. Since a passive transformation sends coherent states to coherent states, each insertion is undone by the reproducing property of the family, and \eqref{eq:pod-nested} collapses to the three-mode Gaussian kernel of the whole chain. The same conclusion follows from the Heisenberg action of the chain on the creation operator,
\begin{equation}
a_0^\dagger \;\longmapsto\; u\,a_0^\dagger+v\,a_1^\dagger+w\,a_2^\dagger ,
\qquad
u=t_1t_2,\quad v=-\sqrt{1-t_1^2},\quad w=-t_1\sqrt{1-t_2^2},
\label{eq:pod-uvw}
\end{equation}
with $u^2+v^2+w^2=1$. The signs of $v$ and $w$ follow the convention of Sec.~\ref{sec:beamsplitter} and may be removed by phase shifters on the auxiliary modes.

\begin{figure}[t]
\centering
\includegraphics[width=0.85\textwidth]{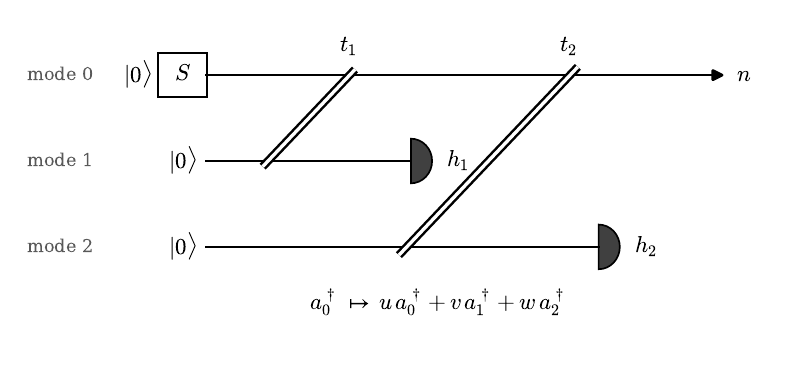}
\caption{The chain \eqref{eq:pod-chain}. A single-mode squeezed vacuum enters mode $0$; two beam splitters of amplitude transmission $t_1$ and $t_2$ couple it to the auxiliary modes $1$ and $2$, which start in vacuum and end on photon-number-resolving detectors registering $h_1$ and $h_2$. Detection heralds the signal mode, whose Fock amplitudes are indexed by $n$. The whole network acts on the creation operator of mode $0$ as in \eqref{eq:pod-uvw}.}
\label{fig:podoshvedov}
\end{figure}

Expanding the squeezed vacuum in the number basis and applying the multinomial theorem to \eqref{eq:pod-uvw} gives the exact heralded amplitude in closed form. With $\tau=\tanh s$ and $L=n+h_1+h_2$,
\begin{equation}
\boxed{\;
\bra{n,h_1,h_2}W\ket{0,0,0}
=\begin{cases}
\dfrac{(-\tau/2)^{L/2}}{\sqrt{\cosh s}}\;\dfrac{L!}{(L/2)!}\;
\dfrac{u^{n}v^{h_1}w^{h_2}}{\sqrt{n!\,h_1!\,h_2!}}, & L \text{ even},\\[10pt]
0, & L \text{ odd}.
\end{cases}\;}
\label{eq:pod-amplitude}
\end{equation}
The parity selection rule is the pair-production structure of the source: the squeezed vacuum contains only even photon numbers, and a passive network preserves the total count.

\paragraph{Heralding probability.} Summing $|\bra{n,h_1,h_2}W\ket{0,0,0}|^2$ over the signal index $n$ gives the probability of the detection pattern $(h_1,h_2)$. With
\begin{equation}
y_0=\frac{\tau}{2},\qquad y=u^2y_0,\qquad Z(y)=\bigl(1-4y^2\bigr)^{-1/2} ,
\end{equation}
the sum is a derivative of the generating function $Z$,
\begin{equation}
P_{h_1,h_2}
=\frac{(y_0v^2)^{h_1}(y_0w^2)^{h_2}}{h_1!\,h_2!\,\cosh s}\;Z^{(H)}(y) ,
\qquad H=h_1+h_2 ,
\label{eq:pod-probability}
\end{equation}
which reproduces the probability formula of Ref.~\cite{Podoshvedov2023} obtained there by a different route. Two structural consequences follow directly from \eqref{eq:pod-amplitude} and \eqref{eq:pod-probability}. First, at fixed $H$ the detection pattern is distributed binomially,
\begin{equation}
P_{h_1,h_2}=P_H\,\binom{H}{h_1}\,\omega_1^{h_1}\omega_2^{h_2} ,
\qquad
\omega_1=\frac{v^2}{1-u^2},\quad \omega_2=\frac{w^2}{1-u^2},\quad P_H=\!\!\sum_{h_1+h_2=H}\!\!P_{h_1,h_2} .
\label{eq:pod-binomial}
\end{equation}
Second, the normalised heralded state of the signal mode depends on $h_1$ and $h_2$ only through $H$, since the $(h_1,h_2)$ dependence of \eqref{eq:pod-amplitude} is a factor independent of $n$. This is the property that makes the scheme useful, and here it follows directly from the closed form.

\paragraph{One squeezed source and an arbitrary passive network.} The derivation uses two properties of the chain: the network is passive, and a single mode is squeezed.

\begin{proposition}[Single-source networks]
\label{prop:class}
Let one mode be prepared in a squeezed vacuum $S(s)\ket0$, let the remaining $M-1$ modes be in vacuum, and let the modes pass through a passive network $U\in \mathrm U(M)$ with $u_j=U_{j1}$. Then, for every output pattern $\mathbf n=(n_1,\dots,n_M)$ with $H=\sum_j n_j$,
\begin{equation}
\boxed{\;
\mathcal A_{\mathbf n}
=\begin{cases}
\dfrac{(-\tau/2)^{H/2}}{\sqrt{\cosh s}}\;\dfrac{H!}{(H/2)!}\;
\dfrac{\prod_j u_j^{\,n_j}}{\sqrt{\prod_j n_j!}}, & H \text{ even},\\[10pt]
0, & H \text{ odd}.
\end{cases}\;}
\label{eq:pod-class}
\end{equation}
\end{proposition}

\begin{proof}
Insert a resolution of the identity between the network and the source as in \eqref{eq:pod-nested} and use the coherent-state map of the passive network, or equivalently apply $U a_1^\dagger U^\dagger=\sum_j u_j a_j^\dagger$ to the number-basis expansion of $S(s)\ket0$ and read off the multinomial coefficient.
\end{proof}

Equation \eqref{eq:pod-class} separates the two ingredients: the network enters only through the monomial $\prod_j u_j^{n_j}$, while the entire dependence on the source is the single factor in front, common to all patterns with the same $H$. Amplitudes of this type are also reachable from the loop-hafnian formulation \cite{Hamilton2017,Bjorklund2019,Quesada2019}, where they correspond to a rank-one matrix; the derivation here needs only the kernels of the individual factors.

\paragraph{Other routes.} As a Gaussian circuit with a single squeezed input, this chain is also within reach of the loop-hafnian formulation \cite{Hamilton2017,Bjorklund2019,Quesada2019}, where the relevant matrix has rank one, and of the recursive Fock constructions of Ref.~\cite{MiattoQuesada2020}. What the insertion procedure adds is that the amplitude is assembled from the kernels of the three factors, with the passive network entering only through the monomial in \eqref{eq:pod-class}.

\subsection{Applications to Quantum Hamiltonians}
\label{sec:hamiltonians}

This subsection extracts Fock-basis matrix elements of normally ordered Hamiltonians. Coherent states are the natural bridge between the quantum and classical descriptions of a bosonic mode, and the diagonal matrix element $H(\beta^*,\beta)=\bra\beta H\ket\beta$ of a Hamiltonian --- its \emph{normal-ordered symbol} --- is the starting point for semiclassical approximations, phase-space methods, and the computation of expectation values in driven or dissipative dynamics. It is now shown that \eqref{eq:MI3b} extracts the full set of \emph{Fock-basis} matrix elements $\bra n H\ket\alpha$ directly from this symbol, for any normal-ordered Hamiltonian, in a single closed formula.

\paragraph{General formula.}
Let $H(a^\dagger,a)=\sum_{j,l}h_{jl}\,(a^\dagger)^ja^l$ be normal-ordered (creation operators to the left), and group its terms by the power of $a$
\begin{equation}
H(a^\dagger,a) = \sum_l f_l(a^\dagger)\,a^l ,
\qquad
f_l(z) := \sum_j h_{jl}\,z^j .
\end{equation}
The annihilation factors stand to the right of the creation factors and therefore act first. Since $a^l\ket\alpha=\alpha^l\ket\alpha$, the action of $H$ on a coherent state reduces immediately to a sum of \emph{creation-operator-only} functions
\begin{equation}
H(a^\dagger,a)\ket\alpha = \sum_l \alpha^l\,f_l(a^\dagger)\ket\alpha ,
\end{equation}
and each term is now exactly of the type covered by \eqref{eq:MI3b}, since $f_l$ is independent of $\alpha$: inserting the resolution of the identity and using $\bra\beta f_l(a^\dagger)\ket\alpha = f_l(\beta^*)\braket\beta\alpha$ gives
\begin{equation}
\bra n f_l(a^\dagger)\ket\alpha = \frac{e^{-|\alpha|^2/2}}{\sqrt{n!}}\,f_l(\partial_\alpha)\,\alpha^n .
\end{equation}
Summing over $l$ yields the general result
\begin{equation}
\boxed{\;
\bra n H(a^\dagger,a)\ket\alpha
= \frac{e^{-|\alpha|^2/2}}{\sqrt{n!}}\sum_l \alpha^l\, f_l(\partial_\alpha)\,\alpha^n
= \frac{n!\,e^{-|\alpha|^2/2}}{\sqrt{n!}}\sum_{j\le n,\,l}\frac{h_{jl}}{(n-j)!}\,\alpha^{n-j+l} ,
\;}
\label{eq:H-general}
\end{equation}
valid for any normal-ordered $H$ whose coefficients $h_{jl}$ decay fast enough for the double sum to converge.

\paragraph{Check 1: the harmonic oscillator.}
For $H=\hbar\omega(N+\tfrac12)$, the only nonzero coefficients are $h_{11}=\hbar\omega$ and $h_{00}=\hbar\omega/2$. Equation \eqref{eq:H-general} gives, for $n\ge1$
\begin{equation}
\bra n H\ket\alpha
= n!\,e^{-|\alpha|^2/2}\Bigl[\frac{\hbar\omega}{(n-1)!}+\frac{\hbar\omega/2}{n!}\Bigr]\frac{\alpha^n}{\sqrt{n!}}
= \hbar\omega\bigl(n+\tfrac12\bigr)\,e^{-|\alpha|^2/2}\,\frac{\alpha^n}{\sqrt{n!}} ,
\end{equation}
recovering the elementary result directly, as it must, since $H$ is diagonal in the Fock basis and $\bra n H\ket\alpha=\hbar\omega(n+\tfrac12)\braket n\alpha$ trivially.

\paragraph{Check 2: a non-diagonal (two-photon) Hamiltonian.}
A sharper test is a Hamiltonian that mixes Fock states, such as the two-photon (degenerate parametric) Hamiltonian
\begin{equation}
H = \hbar\omega\bigl(N+\tfrac12\bigr) + \kappa\bigl(a^{\dagger2}+a^2\bigr) ,
\end{equation}
familiar from the parametric-oscillator and squeezing literature, with coupling constant $\kappa$, whose nonzero coefficients are $h_{11}=\hbar\omega$, $h_{00}=\hbar\omega/2$, $h_{20}=\kappa$, $h_{02}=\kappa$. Equation \eqref{eq:H-general} gives directly, for $n\ge2$
\begin{equation}
\bra n H\ket\alpha
= n!\,e^{-|\alpha|^2/2}\left[\frac{\hbar\omega}{(n-1)!}\alpha^n+\frac{\hbar\omega/2}{n!}\alpha^n+\frac{\kappa}{(n-2)!}\alpha^{n-2}+\frac{\kappa}{n!}\alpha^{n+2}\right]\frac{1}{\sqrt{n!}} .
\end{equation}
This Hamiltonian couples $\ket n$ to $\ket{n\pm2}$, so $\bra n H\ket\alpha$ involves the Fock components $n-2$, $n$ and $n+2$ of $\ket\alpha$. The expression agrees with a direct matrix--vector product in a truncated Fock space for representative parameter values.

\paragraph{Remark.} Equation \eqref{eq:H-general} is a consequence of \eqref{eq:MI3b}, as opposed to \eqref{eq:MI3a}: grouping $H$ by powers of $a$ isolates in each term an $\alpha$-independent function $f_l(a^\dagger)$, to which \eqref{eq:MI3b} applies.

%%%%%%%%%%%%%%%%%%%%%%%%%%%%%%%%%%%%%
\subsection{Numerical checks}
\label{sec:numerics}

This subsection compares the closed forms with independent numerical computations. Operator exponentials were computed by matrix exponentiation in a truncated Fock space. Special functions and series were evaluated with arbitrary-precision arithmetic. All deviations quoted below are maximal absolute differences over the stated ranges of parameters and indices.

\paragraph{Kerr chain.} For $(p,q)\in\{(1,2),(1,3),(2,5),(3,4)\}$, $\alpha=0.7-0.4i$, $(r,\theta)=(0.6,0.3)$ and $n\le7$, Eq.~\eqref{eq:kerrclosed} agrees with the direct computation of $S(\xi)e^{i\chi N^2}\ket\alpha$ to within $1.1\times10^{-15}$. The numbers of non-zero Gauss coefficients are $2,3,5,4$ out of $L=4,6,10,8$, respectively, as Proposition~\ref{prop:kerr}(iii) requires.

\paragraph{Two-mode network.} For $(r,\theta,\varphi)=(0.6,0.3,0.4)$, $\alpha=0.5+0.2i$, $\lambda=0.3-0.4i$ and $n,m\le4$, Eq.~\eqref{eq:composite-final-corrected} agrees with the truncated two-mode computation (cutoff $34$ per mode) to $7.8\times10^{-16}$.

\paragraph{Kerr phase in a two-mode network.} For $(r,\theta,\varphi)=(0.35,0.7,0.5)$, $\alpha=0.6-0.3i$, $\lambda=-0.4+0.5i$ and $n,m\le4$, the finite sums \eqref{eq:kerrnet} were compared with the direct computation of $W\,e^{i\chi N_1^2}\ket{\alpha,\lambda}$ in a truncated two-mode Fock space with cutoff $34$ per mode. Table~\ref{tab:kerrnet} gives the maximal deviations; both parities of $pq$ are included.

\begin{table}[t]
\centering
\caption{Maximal deviation, over $n,m\le4$, between the direct computation of $\bra{n,m}W\,e^{i\chi N_1^2}\ket{\alpha,\lambda}$ and the $q$-term sum \eqref{eq:kerrnet}, for $\chi=\pi p/q$.}
\label{tab:kerrnet}
\begin{tabular}{ccccccc}
\hline
$p/q$ & $1/2$ & $1/3$ & $2/3$ & $1/4$ & $1/5$ & $3/5$\\
terms & 2 & 3 & 3 & 4 & 5 & 5\\
deviation & $2.0\times10^{-16}$ & $4.3\times10^{-16}$ & $6.2\times10^{-16}$ & $5.7\times10^{-16}$ & $2.4\times10^{-16}$ & $1.6\times10^{-15}$\\
\hline
\end{tabular}
\end{table}

\paragraph{Heralded cat-state circuit.} For $t_1=0.80$, $t_2=0.77$ and $s=0.6465$, the closed form \eqref{eq:pod-amplitude} agrees with a direct three-mode truncated Fock simulation of the chain \eqref{eq:pod-chain} (cutoff $26$ per mode) to $8.3\times10^{-17}$ for all patterns with $n,h_1,h_2\le6$. The probability \eqref{eq:pod-probability} agrees with the direct sum over $n$ to a relative $8.8\times10^{-12}$ at $H=4$, the binomial factorisation \eqref{eq:pod-binomial} holds to $1.1\times10^{-15}$, and the normalised heralded states obtained from the different splittings of $H=4$ coincide to $4.4\times10^{-16}$. For random passive networks with $M=2,3,4$ modes, Proposition~\ref{prop:class} agrees with the corresponding simulation to $1.6\times10^{-16}$.

\paragraph{Truncated Fock simulation.} A common way to simulate such a chain is to truncate the Fock space at dimension $D$, build the squeezer as the matrix exponential of the truncated generator, and apply the diagonal Kerr factor \cite{Provaznik2022}. For $\chi=\pi/3$, $\xi=0.4$ and $n\le9$, the three-term sum \eqref{eq:kerrclosed} agrees with this computation at $D=150$ to $1.1\times10^{-15}$ for $\alpha=2$ and to $3.9\times10^{-15}$ for $\alpha=3$. Table~\ref{tab:trunc} gives the maximal deviation of the truncated computation from the sum at smaller $D$. The mean photon numbers are $5.5$ and $12.2$, and the cutoff needed to reach the double-precision floor grows with them.

\begin{table}[t]
\centering
\caption{Maximal deviation, over $n\le9$, of the truncated Fock simulation of $S(\xi)e^{i\pi N^2/3}\ket\alpha$ from the three-term sum \eqref{eq:kerrclosed}, as a function of the cutoff $D$; $\xi=0.4$.}
\label{tab:trunc}
\begin{tabular}{ccccccc}
\hline
$D$ & 12 & 16 & 20 & 30 & 40 & 60\\
\hline
$\alpha=2$ & $7.6\times10^{-2}$ & $5.2\times10^{-3}$ & $1.9\times10^{-4}$ & $7.4\times10^{-9}$ & $8.0\times10^{-14}$ & $9.4\times10^{-16}$\\
$\alpha=3$ & $2.7\times10^{-1}$ & $8.1\times10^{-2}$ & $9.8\times10^{-3}$ & $2.3\times10^{-6}$ & $7.7\times10^{-11}$ & $3.9\times10^{-15}$\\
\hline
\end{tabular}
\end{table}

\section{Discussion and Conclusion}
\label{sec:discussion}

This section places the procedure among established techniques, summarises its outcome, examines the choice of resolution and states its scope.

\subsection{Reach of the established techniques}
\label{sec:reach}

This subsection compares the reach of the procedure with that of established techniques. The insertion of a resolution of the identity is the computational core of several established techniques, and most of the examples above are within reach of more than one of them. Table~\ref{tab:reach} summarises the reach of each technique, with $\checkmark$ for a direct route, $\circ$ for one that needs a reformulation or an extra structural assumption, and -- for a case outside the scope of the technique.

\begin{table}[t]
\centering
\caption{Reach of four established techniques over the examples of this paper: normal-ordering techniques \cite{FanZaidiKlauder1987,Fan2003,FanYuan2010}, the Berezin symbol calculus \cite{Berezin1972}, Lie disentangling and group-orbit coherent states \cite{Truax1985,Perelomov1986}, and the symplectic and loop-hafnian formalism \cite{Hamilton2017,Bjorklund2019,Quesada2019,MiattoQuesada2020}. For the Kerr row, the established techniques apply after the phase-rotation decomposition of Proposition~\ref{prop:kerr}(i), combined with Diophantine approximation at irrational phases \cite{Upreti2026}.}
\label{tab:reach}
\begin{tabular}{lccccc}
\hline
Example & Ordering & Berezin & Lie/group & Symplectic & This paper\\
\hline
Elementary and polynomial functions & $\checkmark$ & $\checkmark$ & $\checkmark$ & -- & $\checkmark$\\
Analytic functions of $a^\dagger$ & $\checkmark$ & $\checkmark$ & $\circ$ & -- & $\checkmark$\\
Single- and two-mode Gaussian unitaries & $\checkmark$ & $\checkmark$ & $\checkmark$ & $\checkmark$ & $\checkmark$\\
Compound two-mode network & $\checkmark$ & $\checkmark$ & $\checkmark$ & $\checkmark$ & $\checkmark$\\
Kerr phase, rational and irrational & $\circ$ & $\circ$ & $\circ$ & $\circ$ & $\checkmark$\\
Heralded three-factor circuit & $\checkmark$ & $\checkmark$ & $\checkmark$ & $\checkmark$ & $\checkmark$\\
Normally ordered Hamiltonians & $\checkmark$ & $\checkmark$ & -- & -- & $\checkmark$\\
\hline
\end{tabular}
\end{table}

Where they apply, the ordering techniques and the Berezin calculus give the results directly, and for the Gaussian rows the symplectic and hafnian route is the most efficient in practice. A normal-ordering step, a group-orbit structure and a Gaussian covariance matrix are each powerful when available, and each fixes the class of operators to which the technique applies. The insertion procedure covers every row with the same hypotheses: a resolution of the identity, the domain conditions of Sec.~\ref{sec:assumptions}, and kernels for the individual factors.

The entries of Table~\ref{tab:reach} have the following origins.

\emph{Gaussian rows.} In the rows of single- and two-mode Gaussian unitaries, the compound network and the heralded circuit, every factor is a Gaussian unitary and every intermediate state is Gaussian. Each technique then applies directly: the normally ordered forms and covariant symbols are Gaussian, the generators lie in a finite-dimensional Lie algebra, and the symplectic matrices compose. The last column follows in every row from Proposition~\ref{cor:single} and Theorem~\ref{thm:telescope} with the coherent-state resolution, under the assumptions verified in Sec.~\ref{sec:examples}.

\emph{Elementary, polynomial and analytic functions.} The operators $a$, $a^\dagger$, $N$ and the displacement belong to the oscillator algebra spanned by $a$, $a^\dagger$, $N$ and $I$, whose group orbit through the vacuum is the family of Glauber coherent states; polynomials in these generators act on the coherent-state labels as differential operators, which gives the Lie/group entry $\checkmark$ in the first row. An analytic function $f(a^\dagger)$ is the exponential of an element of a finite-dimensional Lie algebra only for special $f$, such as $e^{za^\dagger}$ or $e^{za^{\dagger2}}$. A general $f$ enters disentangling through its Taylor expansion or through a superposition of the exponentials $e^{za^\dagger}$, a reformulation that gives the entry $\circ$. The symplectic formalism describes Gaussian operations, generated by operators of at most quadratic order, through symplectic matrices and displacement vectors. Within the first row it covers the displacement, and the entry -- refers to the polynomials $(a^\dagger)^ka^l$; these and a general $f(a^\dagger)$ map Gaussian states to non-Gaussian ones, such as photon-added and photon-subtracted states, which lie outside the covariance-matrix description.

\emph{Kerr phase.} The generator $N^2$ is quartic in $a$ and $a^\dagger$. Its normally ordered form follows from Newton's forward-difference series,
\begin{equation}
e^{i\chi N^2}=\sum_{k=0}^{\infty}\frac{(\Delta^kf)(0)}{k!}\,(a^\dagger)^ka^k,\qquad f(n)=e^{i\chi n^2},\quad (\Delta f)(n)=f(n+1)-f(n),
\label{eq:kerr-newton}
\end{equation}
since $(a^\dagger)^ka^k$ acts on $\ket n$ as the falling factorial $n!/(n-k)!$. The sequence $f$ is a polynomial in $n$ only for $\chi\in2\pi\mathbb Z$, so for every other phase the series is infinite. The covariant symbol is likewise the theta-type series
\begin{equation}
\bra\beta e^{i\chi N^2}\ket\beta=e^{-|\beta|^2}\sum_{k=0}^{\infty}\frac{|\beta|^{2k}}{k!}\,e^{i\chi k^2}.
\end{equation}
Nested commutators of $N^2$ with $a$, $a^\dagger$, $a^2$ and $a^{\dagger2}$ have unbounded degree in the ladder operators, so the Lie algebra they generate is infinite-dimensional, and the Kerr phase is a non-Gaussian unitary. Each of the four techniques therefore reaches this row through the decomposition of Proposition~\ref{prop:kerr}(i): at $\chi=\pi p/q$ the ordered form and the covariant symbol become sums of $q$ phase rotations, the latter equal to $\sum_jg_j\,e^{(\omega^j-1)|\beta|^2}$; the rotations $e^{2\pi ijN/L}$ belong to the oscillator algebra; and each term of the decomposed chain is Gaussian for the symplectic and hafnian route. At irrational $\chi/\pi$ the same route uses Diophantine approximation \cite{Upreti2026}. This reformulation gives the entries $\circ$, while Eq.~\eqref{eq:kerr-result} holds for every real $\chi$.

\emph{Normally ordered Hamiltonians.} This row concerns the matrix elements $\bra nH\ket\alpha$ of a Hamiltonian $H=\sum_{j,l}h_{jl}(a^\dagger)^ja^l$ of arbitrary degree. Its normally ordered form coincides with its covariant symbol $H(\beta^*,\beta)$, which gives the entries $\checkmark$ for ordering and for the Berezin calculus. Lie disentangling and the symplectic formalism address the evolution operator $e^{-iHt}$, which for a quadratic $H$ they describe through an element of $\mathfrak{su}(1,1)$ or through a symplectic matrix. The quantity in this row is linear in $H$, and a Hamiltonian of degree above two lies outside every finite-dimensional Lie algebra of the ladder operators and outside the Gaussian class, which gives the entries --.

\subsection{What the framework adds}

This subsection summarises the outcome of the procedure. Theorem~\ref{thm:telescope} licenses sequential insertions between the factors of a chain of admissible operators, with the analytic conditions stated explicitly. The Kerr-squeezed coherent overlap (Sec.~\ref{sec:kerr}), the compound Gaussian network (Sec.~\ref{sec:composite}) and the three-factor heralded circuit (Sec.~\ref{sec:podoshvedov}) are evaluated from the kernels of their factors, with the composite operator kept in factored form. The squeezed-then-mixed Fock amplitudes yield the bivariate Hermite polynomial \eqref{eq:bivariate-hermite}, which displays the self-squeezing terms generated in both output modes.

\subsection{Choice of resolution}
\label{sec:choice}

This subsection examines the freedom in the choice of resolution at each boundary of a chain. For unitary operators $U_j$, the family $\{U_j\ket\gamma\}$ resolves the identity whenever $\{\ket\gamma\}$ does, and inserting rotated families at the two boundaries of the factor $\mathcal F_j$ replaces its kernel by
\begin{equation}
\bra{\gamma'}U_{j+1}^\dagger\,\mathcal F_j\,U_j\ket{\gamma},
\end{equation}
the kernel of the refactorized chain with factors $U_{j+1}^\dagger\mathcal F_jU_j$. The value of each matrix element is independent of this choice, while the form of the resulting expression and its computational cost depend on it. The examples of Sec.~\ref{sec:examples} and the comparisons of Sec.~\ref{sec:kerr} follow a selection rule. The efficient family at a boundary is one in which the adjacent factor acts diagonally, as the Fock basis for the Kerr phase and the spectral resolution of the position quadrature for the cubic phase gate; one that the factor maps onto members of the same family, as coherent states for Gaussian unitaries; or one whose members the factor annihilates, as eigenstates of $a^2$ for two-photon-driven Kerr oscillators, whose even and odd cat states are exact eigenstates \cite{PuriBoutinBlais2017}. Inserting the Fock basis at the beam-splitter boundary of the network of Sec.~\ref{sec:composite} gives a correct finite sum at four to five times the cost of Eq.~\eqref{eq:bivariate-hermite}, in line with this rule. For the cubic phase gate on a pure Gaussian input, the position resolution gives the whole Fock profile from two Airy-function values, whereas a Fock resolution gives a truncated sum whose length grows with the squeezing of the input \cite{Asadi2026cubic}; Eq.~\eqref{eq:kerr-cubic} combines the two rules, with coherent states at the boundary between the Kerr phase and the cubic gate.

Group-orbit coherent states of the squeezing algebra, labelled by points $\zeta$ of the unit disk, resolve the identity in the sectors with Bargmann index $k>1/2$ \cite{Perelomov1986}. The single-mode squeezed vacua belong to $k=1/4$. A rotation-invariant resolution in this sector would require a positive measure in the variable $t=|\zeta|^2\in[0,1]$ with the moments
\begin{equation}
s_m=\frac{m!\,\Gamma(2k)}{\Gamma(2k+m)},\qquad \frac{s_{m+1}}{s_m}=\frac{m+1}{m+2k},
\end{equation}
which increase with $m$ for $k<1/2$, whereas the moments of a positive measure on $[0,1]$ decrease; averaging over rotations extends this conclusion to every positive measure on the disk. The eigenstates of $a^2$ give a positive resolution for every $k>0$ \cite{BarutGirardello1971}, and for $k=1/4$ they are the even cat states. A quantitative form of the selection rule, with the cost defined as the number of terms required for a prescribed accuracy, is a natural direction for further work.

\subsection{Scope}

This subsection states the range of validity of the results. Assumptions~4 and 5 are verified here for coherent-state resolutions acting on Gaussian-dominated symbols, and each new operator family and resolution of the identity calls for its own verification; the scalar identities of Proposition~\ref{cor:single} hold for every vector in the domain of $\mathcal F$. The theorem treats vectors $\mathcal F\ket\psi$ of finite norm in $\mathcal H$. The matrix elements of Sec.~\ref{sec:examples} require less than these vector-level hypotheses: a single insertion between two factors holds for matrix elements whenever the bra lies in the domain of the adjoint of the left factor, the resolving vectors lie in the domain of that factor, and the ket lies in the domain of the right factor, and it then gives the product in the sense of sesquilinear forms. This weak reading, its relation to the Bochner conditions used here, and an explicit unbounded chain in which the two readings separate are developed in Ref.~\cite{AsadiJPA2026}. Theorem~\ref{thm:telescope} holds for an arbitrary chain of closed operators satisfying Assumptions~4 and 5 at each stage, and the closed-form evaluation of the resulting multi-fold integral uses increasingly elaborate special functions as the chain grows, as seen in passing from the two-mode squeezer of Sec.~\ref{sec:2mode-squeeze} to the bivariate Hermite polynomial of Sec.~\ref{sec:composite}.

\subsection{Prospective applications}

This subsection outlines directions in which the results apply. The nested identity \eqref{TelescopeN} represents the exact discrete-time skeleton of a Feynman-type path-integral relation. Taking the continuum limit of \eqref{TelescopeN} for a Hamiltonian evolution operator factored into infinitesimal time steps connects this framework directly to coherent-state path-integral quantization in the spirit of Klauder's original program \cite{Klauder1963,Klauder1963_2}. For Gaussian networks, the closed forms obtained by kernel composition are explicit loop hafnians, as Eq.~\eqref{eq:lhaf} shows for two modes, and connect the procedure to the amplitudes of Gaussian boson sampling \cite{Madsen2022}.

\paragraph{Applications of the Kerr and two-mode results.} In a Kerr-tunable superconducting resonator, a displaced state evolves under a Kerr gate for a time $\tau_0/m$ and becomes an $m$-component cat state, realised for $m=2,3,4$ \cite{He2023}. This is the case $\chi=\pi/m$ of Proposition~\ref{prop:kerr}, and a subsequent squeezer or displacement gives the Fock amplitudes as sums of $m$ terms, by \eqref{eq:kerrclosed} and \eqref{eq:kerr-disp}; Corollary~\ref{cor:kerrnet} extends them to a Kerr phase embedded in the network of Sec.~\ref{sec:composite}. In protocols that generate large Fock states by repeated Kerr evolutions and displacements \cite{Meher2026}, the amplitudes at rational Kerr phases are sums of $q$ terms, and Proposition~\ref{prop:kerrapprox} bounds the error at the remaining phases. Simulations of photonic circuits with non-Gaussian elements are commonly performed in a truncated Fock space, where the numerical errors have been analysed for the displacement operator \cite{Provaznik2022}, while recursive constructions give the exact Fock matrix elements of Gaussian transformations to any dimension \cite{MiattoQuesada2020}. The finite sums supply exact test vectors for Kerr-plus-Gaussian chains in such codes (Table~\ref{tab:trunc}). For two modes, amplitudes of the type obtained in Sec.~\ref{sec:composite} enter the heralded preparation of non-Gaussian states with photon-number-resolving detectors \cite{Su2019,Quesada2019NG} and the Franck--Condon factors of vibronic spectra \cite{Huh2015}. The general multimode case is covered by the hafnian methods cited above; the contribution here is the derivation by kernel composition.

In summary, the paper formulates a familiar computational step as an explicit and analytically controlled procedure, extends it to chains of operators, and uses it to produce exact matrix elements for composite operations, each verified numerically.
%%%%%%%%%%%%%%%%%%%%%%%
\section*{Author Declarations}

\subsection*{Conflict of Interest}

The author has no conflicts to disclose.

\subsection*{Author Contributions}

The author is solely responsible for the conceptualization, mathematical development, analysis, and writing of this work.

\section*{Data Availability}

The numerical values reported in Sec.~\ref{sec:numerics} were generated with scripts written for this work. The scripts and the generated data are available from the author upon reasonable request.

\end{document}